\documentclass[english]{article}
\usepackage{lmodern}

\usepackage[T1]{fontenc}
\usepackage[latin9]{inputenc}
\usepackage{geometry}
\usepackage{babel}
\usepackage{verbatim}
\usepackage{varioref}
\usepackage{amsmath}
\usepackage{amsthm}
\usepackage{amssymb}
\usepackage{stmaryrd}
\usepackage{graphicx}
\usepackage[unicode=true,pdfusetitle,
 bookmarks=true,bookmarksnumbered=false,bookmarksopen=false,
 breaklinks=true,pdfborder={0 0 0},pdfborderstyle={},backref=page,colorlinks=false]
 {hyperref}

\makeatletter
\numberwithin{equation}{section}
\numberwithin{figure}{section}
\numberwithin{table}{section}
\theoremstyle{definition}
\newtheorem{defn}{\protect\definitionname}
\theoremstyle{plain}
\newtheorem{lem}{\protect\lemmaname}
\ifx\proof\undefined\
  \newenvironment{proof}[1][\proofname]{\par
    \normalfont\topsep6\p@\@plus6\p@\relax
    \trivlist
    \itemindent\parindent
    \item[\hskip\labelsep
          \scshape
      #1]\ignorespaces
  }{%
    \endtrivlist\@endpefalse
  }
  \providecommand{\proofname}{Proof}
\fi
\theoremstyle{plain}
\newtheorem{thm}{\protect\theoremname}
\newenvironment{lyxcode}
	{\par\begin{list}{}{
		\setlength{\rightmargin}{\leftmargin}
		\setlength{\listparindent}{0pt}% needed for AMS classes
		\raggedright
		\setlength{\itemsep}{0pt}
		\setlength{\parsep}{0pt}
		\normalfont\ttfamily}%
	 \item[]}
	{\end{list}}
\providecommand*{\code}[1]{\texttt{#1}}

\makeatother

\usepackage{listings}
\providecommand{\definitionname}{Definition}
\providecommand{\lemmaname}{Lemma}
\providecommand{\theoremname}{Theorem}

\begin{document}
\title{Mutual Coinduction}
\author{Moez A. AbdelGawad\\
\code{moez@cs.rice.edu}}
\maketitle
\begin{abstract}
In this paper%
\begin{comment}
 (which hitherto is an ongoing work-in-progress)
\end{comment}
{} we present mutual coinduction as a dual of mutual induction and also
as a generalization of standard coinduction. In particular, we present
a precise formal definition of mutual induction and mutual coinduction.
In the process we present the associated mutual induction and mutual
coinduction proof principles, and we present the conditions under
which these principles hold.

In spite of some mention of mutual (co)induction in research literature,
but the formal definition of mutual (co)induction and the proof of
the mutual (co)induction proof principles we present here seem to
be the first such definition and proof. As such, to the best of our
knowledge, it seems our work is the first to point out that, unlike
the case for standard (co)induction, monotonicity of generators seems
not sufficient for guaranteeing the existence of least and greatest
simultaneous fixed points in complete lattices, and that continuity
on the other hand is sufficient for guaranteeing their existence.\footnote{This paper has been updated so as to not require the continuity of
generators but only require their monotonicity. See Appendix~\ref{sec:Requiring-Only-Monotonicity}.
A full revision of the paper to reflect the relaxed requirement is
currently underway.}

In the course of our presentation of mutual coinduction we also discuss
some concepts related to standard (also called direct) induction and
standard coinduction, as well as ones related to mutual (also called
simultaneous or indirect) induction. During the presentation we purposely
discuss particular standard concepts (namely, fixed points, least
and greatest fixed points, pre-fixed points, post-fixed points, least
pre-fixed points, and greatest post-fixed points) so as to help motivate
the definitions of their more general counterparts for mutual/simul\-taneous/in\-direct
(co)induction (namely, simultaneous fixed points, least simultaneous
fixed points, greatest simultaneous fixed points, least simultaneous
pre-fixed points and greatest simultaneous post-fixed points). Greatest
simultaneous post-fixed points, in particular, will be abstractions
and models of mathematical objects (\emph{e.g.}, points, sets, types,
predicates, etc.) that are defined mutually-coinductively.
\end{abstract}

\section{\label{sec:Introduction}Introduction}

Induction and coinduction, henceforth called \emph{standard }(\emph{co})\emph{induction},
present mathematical models for (standard) recursive definitions (also
called \emph{direct} recursive or \emph{self}-recursive definitions).
In the same way, mutual induction and mutual coinduction, henceforth
called \emph{mutual }(\emph{co})\emph{induction}, present mathematical
models for mutually-recursive definitions, which are sometimes also
called \emph{indirect} recursive definitions or \emph{simultaneous}
recursive definitions.

In other concurrent work we present, first, some practical motivations
from programming languages theory for defining mutual (co)induction~\cite{AbdelGawad2019e}
(included in Appendix~\ref{sec:Motivations-from-PLT} of this paper
temporarily), then we present some intuitions for their mathematical
definitions as well as examples of their use in~\cite{AbdelGawad2019c},
and we also present possible formulations of mutual (co)induction
in other mathematical disciplines in~\cite{AbdelGawad2019d}. In
that concurrent work we conclude that: if mutually-recursive functional
programs can be reasoned about mathematically, then also imperative
and object-oriented programs (even the worst such programs) can be
reasoned about mathematically. Interested readers are invited to check
this concurrent work.

This paper is structured as follows. We first motivate the formal
definition of mutual (co)induction by presenting in~$\mathsection$\ref{subsec:Standard-(Co)Induction}
the formal definitions of standard (co)induction, in the context of
order theory. Then we present in~$\mathsection$\ref{sec:Mutual-(Co)Induction}
the order-theoretic definitions of mutual (co)induction and we present
the mutual (co)induction proof principles. (We present proofs for
the lemmas and theorems of~$\mathsection$\ref{sec:Mutual-(Co)Induction}
in Appendix~\ref{sec:Theorems-and-Proofs}.) We briefly discuss some
related work in~$\mathsection$\ref{sec:Related-Work}, then we conclude
and discuss some possible future work in~$\mathsection$\ref{sec:Conclusion-and-FW}.

Our work here is a followup on our earlier introductory work in~\cite{AbdelGawad2018f,AbdelGawad2019a,AbdelGawad2019b}.
In addition to the practical motivations mentioned above (whose interest
in we started in~\cite{AbdelGawad2017c}), the work presented here
has also been motivated by that earlier introductory work.

\section{\label{subsec:Standard-(Co)Induction}Standard Induction and Standard
Coinduction}

The formulation of standard induction and standard coinduction, and
of related concepts, that we present here is a summary of the formulation
presented in~\cite[$\mathsection$2.1]{AbdelGawad2018f}.

Let $\leq$ (`is less than or equal to') be an ordering relation
over a set $\mathbb{O}$ and let $F:\mathbb{O\rightarrow\mathbb{O}}$
be an endofunction over $\mathbb{O}$ (also called a self-map over
$\mathbb{O}$, \emph{i.e.}, a function whose domain and codomain are
the same set, thus mapping a set into itself%
\begin{comment}
\footnote{We are focused on unary functions in this paper because we are interested
in discussing fixed points and closely-related concepts, to which
multi-arity makes little difference. Note that a binary function $F:\mathbb{O\times O\rightarrow O}$
can be transformed into an equivalent unary one $F':\mathbb{O\rightarrow\left(O\rightarrow O\right)}$
via the technique of ``currying'' (also known in logic as \emph{exportation}).
By iterating currying, it can be applied to multi-ary functions, \emph{i.e.},
functions with any finite arity greater than two. Currying does preserve
monotonicity/variance, and it seems currying is applicable to all
fields of interest to us in this paper since, in each field, the objects
of that field---\emph{i.e.}, posets, power sets, types, etc.---and
the ``morphisms/arrows'' between these objects seem to form what
is called (in category theory) a \emph{closed monoidal category}.}
\end{comment}
).

Given a point $P\in\mathbb{O}$, the point $F\left(P\right)$ is called
the \emph{$F$-image} of $P$.

A point $P\in\mathbb{O}$ is called a \emph{pre-fixed point} of $F$
if its $F$-image is less than or equal to it, \emph{i.e.}, if 
\[
F\left(P\right)\leq P.
\]

Dually, a point $P\in\mathbb{O}$ is called a \emph{post-fixed point}
of $F$ if it is less than or equal to its $F$-image, \emph{i.e.},
if 
\[
P\leq F\left(P\right).
\]

A point $P\in\mathbb{O}$ is called a \emph{fixed }point of $F$ if
it is equal to its $F$-image, \emph{i.e.}, if 
\[
P=F\left(P\right).
\]
 (A fixed point of $F$ is simultaneously a pre-fixed point of $F$
and a post-fixed point of $F$.)\medskip{}

Now, if $\leq$ is a complete lattice over $\mathbb{O}$ (\emph{i.e.},
if $\leq$ is an ordering relation where meets $\wedge$ and joins
$\vee$ of \emph{all} subsets of $\mathbb{O}$ are guaranteed to exist
in $\mathbb{O}$) and if $F$ is a monotonic\emph{ }endofunction over
$\mathbb{O}$ ($F$ is then called a \emph{generator}), then the \emph{least
pre-fixed point }of $F$, called $\mu_{F}$, exists in $\mathbb{O}$,
by the Knaster-Tarski Fixed Point Theorem~\cite{Knaster1928,Tarski1955},
$\mu_{F}$ is also the \emph{least fixed point} (\emph{lfp}) of $F$,
and the \emph{greatest post-fixed point} of $F$, called $\nu_{F}$,
exists in $\mathbb{O}$ and, again by the Fixed Point Theorem, $\nu_{F}$
is also the \emph{greatest fixed point} (\emph{gfp}) of $F$.

Further, %
\begin{comment}
given that $\mu_{F}$ is the \emph{least} pre-fixed point\emph{ }of
$F$ and $\nu_{F}$ is the \emph{greatest} post-fixed point\emph{
}of $F$, 
\end{comment}
for any element $P\in\mathbb{O}$ we have:
\begin{itemize}
\item (\emph{standard induction})\phantom{co} if $F\left(P\right)\leq P$,
then $\mu_{F}\leq P$,\medskip{}
\\
which, in words, means that if $P$ is a pre-fixed/induc\-tive point
of $F$, then point $\mu_{F}$ is less than or equal to $P$ (since
$\mu_{F}$, by definition, is the least $F$-inductive point), and,
\item (\emph{standard coinduction}) if $P\leq F\left(P\right)$, then $P\leq\nu_{F}$,\medskip{}
\\
which, in words, means that if $P$ is a post-fixed/coinduc\-tive
point of $F$, then point $P$ is less than or equal to point $\nu_{F}$
(since $\nu_{F}$, by definition, is the greatest $F$-coinductive
point).
\end{itemize}

\paragraph{References}

See~\cite{Davey2002,Rom2008,AbdelGawad2018f}.

\section{\label{sec:Mutual-(Co)Induction}Mutual Induction and Mutual Coinduction}

In this section we first present some intuition for mutual (co)induction
as generalizations of standard (co)induction, then we present their
formal definitions and formulate the associated proof principles.
Our illustrated presentation includes a discussion of the continuity
of generators, which seems necessary for proving the existence of
simultaneous fixed points (but see Appendix~\ref{sec:Requiring-Only-Monotonicity}).
In Appendix~\vref{sec:Theorems-and-Proofs} we present proofs for
the lemmas and theorems we present here.

\subsubsection{Intuition}

Intuitively, compared to standard (co)induction which involves one
self-map from a poset to itself, mutual (co)induction involves two
or more mappings between two or more ordered sets (\emph{i.e.}, posets)%
\begin{comment}
 (two/finite/infinite? smallest large)
\end{comment}
. That's all. In this paper, for simplicity, we focus only on the
case involving just two mappings (also called \emph{generators}) between
two posets. As we define it below, our definition of mutual induction
can be extended easily to involve more than two orderings, more than
two underlying sets, and more than two mappings between the ordered
sets. Also, it should be noted that in some practical applications
of mutual (co)induction the two orderings, and their underlying sets,
may happen to be the \emph{same} ordering and set. For proper mutual
(co)induction, however, there has to be at least two mutual generators/mappings
between the two ordered sets.\footnote{Standard induction can then obtained as a special case of mutual induction
by having one ordering and one underlying set \emph{and} also having
one of the two generators---or, more accurately, all but one of the
generators---be the identity function. See also Footnote~\ref{fn:Mut->Stand}.}

Also intuitively, the mutual induction and mutual coinduction proof
principles are expressions of the properties of two points (\emph{i.e.},
elements of the two ordered sets) that are together (\emph{i.e.},
simultaneously) least pre-fixed points of each of the two generators
and of two points that are together greatest post-fixed points of
each of the two generators. As such mutual induction and mutual coinduction
generalize the standard induction and standard coinduction proof principles.
Further, in case the two orderings are complete lattices and the mappings
are continuous (thereby called\emph{ generators} or \emph{mutual generators}),
the two least simultaneous pre-fixed points and the two greatest simultaneous
post-fixed points will also be \emph{mutual fixed points} (sometimes
also called \emph{simultaneous }or \emph{reciprocal }fixed points)
of the generators. (For a glimpse, %
\begin{comment}
TODO: intuitions and a tutorial on mutual (co)induction, see~\cite{AbdelGawad2019c}
\end{comment}
see Figures~\ref{fig:SPreFPs},~\ref{fig:Mult-SPreFPs}, and~\ref{fig:Cont-FG}.)
\begin{figure}
\noindent \begin{centering}
\includegraphics[scale=0.5]{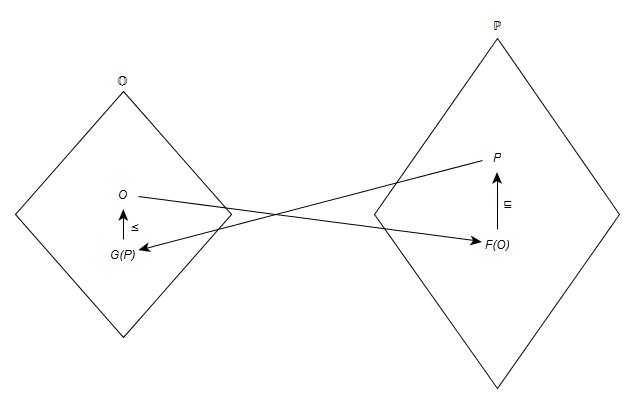}
\par\end{centering}
\caption{\label{fig:SPreFPs}Illustrating simultaneous pre-fixed points, \emph{e.g.},
points $O$ and $P$.}
\end{figure}
\begin{figure}
\noindent \begin{centering}
\includegraphics[scale=0.5]{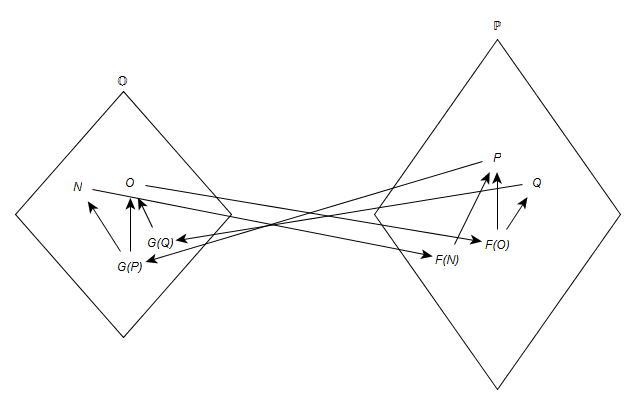}
\par\end{centering}
\caption{\label{fig:Mult-SPreFPs}The points $O$ and $P$, $O$ and $Q$,
and $N$ and $P$ illustrate having multiple simultaneous pre-fixed
points that share some of their component points.}
\end{figure}
\begin{figure}
\noindent \begin{centering}
\includegraphics[scale=0.5]{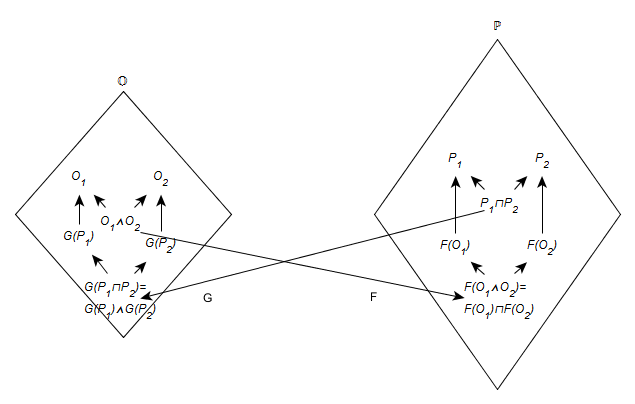}
\par\end{centering}
\caption{\label{fig:Cont-FG}(Meet-)Continuous mutual generators $F$ and $G$.
}
\end{figure}

\subsubsection{Formulation}

Let $\leq$ be an ordering relation over a set $\mathbb{O}$ and let
$\sqsubseteq$ be a second ordering relation over a second set $\mathbb{P}$.
Further, let $F:\mathbb{O\rightarrow\mathbb{P}}$ and $G:\mathbb{P\rightarrow\mathbb{O}}$
be two \emph{mutual} \emph{endofunctions} over $\mathbb{O}$ and $\mathbb{P}$
(also called \emph{indirect }or \emph{reciprocal self-maps} over $\mathbb{O}$
and $\mathbb{P}$, \emph{i.e.}, two functions where the domain of
the first is the same as the codomain of the second \emph{and} vice
versa, such that each of the two posets is mapped into the other).\footnote{As in~\cite{AbdelGawad2018f}, we are focused on unary functions
in this paper because we are interested in discussing fixed points
and closely-related concepts such as induction and coinduction, to
which multi-arity seems to make little difference. For more details,
see the brief discussion on arity and currying in~\cite{AbdelGawad2018f}.} Note that given two mutual endofunctions $F$ and $G$ we can always
compose $F$ and $G$ to get the (standard) endofunctions $G\circ F:\mathbb{O}\rightarrow\mathbb{O}$
and $F\circ G:\mathbb{P}\rightarrow\mathbb{P}$.\footnote{Upon seeing these compositions, readers familiar with category theory
may immediately suspect a possible connection with adjunctions. That
possibility of a connection will increase when readers check the definitions
below, of simultaneous pre-fixed points and simultaneous post-fixed
points. We intend to dwell on the possibility of such a connection
in future work (or in later versions of this paper).}\textsuperscript{,}\footnote{\label{fn:Mut->Stand}Note that if we have $\mathbb{P}=\mathbb{O}$
(the two underlying sets are the same), $\sqsubseteq=\leq$ (with
the same orderings), and if $G=\mathbf{1}$ (or $F=\mathbf{1}$, where
$\mathbf{1}$ is the identity function), then we obtain standard (co)induction
as a special case of mutual (co)induction. In particular, all definitions
presented below will smoothly degenerate to their standard counterparts
(\emph{i.e.}, will correspond to ones for standard (co)induction.
See~$\mathsection$\ref{subsec:Standard-(Co)Induction}).}

Given points $O\in\mathbb{O}$ and $P\in\mathbb{P}$, the point $F\left(O\right)\in\mathbb{P}$
and the point $G\left(P\right)\in\mathbb{O}$ are called the \emph{$F$-image}
of $O$ and the $G$\emph{-image} of $P$, respectively.

Two points $O\in\mathbb{O}$, $P\in\mathbb{P}$ are called \emph{simultaneous}
(or \emph{mutual} or \emph{reciprocal}) \emph{pre-fixed points} of
$F$ and $G$ if the $F$-image of $O$ is less than or equal to $P$
\emph{and} the $G$-image of $P$ is less than or equal to $O$ (that
is, intuitively, if ``the images of the two points are less than
the two points themselves''), \emph{i.e.}, if 
\[
F\left(O\right)\sqsubseteq P\textrm{ and }G\left(P\right)\leq O.
\]
Simultaneous pre-fixed points are also called \emph{mutually-inductive
points} of $F$ and $G$. See Figure~\vref{fig:SPreFPs} for an illustration
of simultaneous pre-fixed points.
\begin{itemize}
\item It should be immediately noted from the definition of simultaneous
pre-fixed points that, generally-speaking, a single point $O\in\mathbb{O}$
can be paired with \emph{more} \emph{than} \emph{one} point $P\in\mathbb{P}$
such that $O$ and $P$ form a single pair of simultaneous pre-fixed
points of $F$ and $G$. Symmetrically, a single point $P\in\mathbb{P}$
can also be paired with more than one point $O\in\mathbb{O}$ to form
such a pair. (See Figure~\vref{fig:Mult-SPreFPs} for an illustration.)
\item As such, two functions $PreFP_{F,G}:\mathbb{O}\rightarrow\wp\left(\mathbb{P}\right)$
and $PreFP_{G,F}:\mathbb{P}\rightarrow\wp\left(\mathbb{O}\right)$
that compute these sets of points in $\mathbb{P}$ and $\mathbb{O}$,
respectively, can be derived from $F$ and $G$. In particular, we
have
\[
PreFP_{F,G}\left(O\right)=\left\{ P\in\mathbb{P}|F\left(O\right)\sqsubseteq P\textrm{ and }G\left(P\right)\leq O\right\} ,\textrm{ and}
\]
\[
PreFP_{G,F}\left(P\right)=\left\{ O\in\mathbb{O}|F\left(O\right)\sqsubseteq P\textrm{ and }G\left(P\right)\leq O\right\} .
\]
Note that for some points $O\in\mathbb{O}$, the set $PreFP_{F,G}\left(O\right)$
can be the empty set $\phi$, meaning that such points are not paired
with any points in $\mathbb{P}$ so as to form simultaneous pre-fixed
points of $F$ and $G$. Symmetrically, the same observation holds
for some points $P\in\mathbb{P}$ and their images $PreFP_{G,F}\left(P\right)$.\footnote{Note that in standard induction---or, more accurately, in the encoding
of standard induction using mutual induction---(ultimately due to
the transitivity of $\leq$ and $\sqsubseteq$) it can be said that
a standard pre-fixed point $O\in\mathbb{O}$ is ``paired'' with
one point, namely itself, to form a pair of simultaneous pre-fixed
points that encodes the standard pre-fixed point. In other words,
standard pre-fixed points in standard induction, when encoded as mutual
induction, correspond to simultaneous pre-fixed points (\emph{e.g.},
using the encoding we mentioned in Footnote~\ref{fn:Mut->Stand}),
and vice versa (\emph{i.e.}, simultaneous pre-fixed points, in such
an encoding, will correspond to standard pre-fixed points).}
\end{itemize}
Dually, two points $O\in\mathbb{O}$, $P\in\mathbb{P}$ are called
\emph{simultaneous} (or \emph{mutual} or \emph{reciprocal}) \emph{post-fixed
points} of $F$ and $G$ if $P$ is less than or equal to the $F$-image
of $O$ \emph{and} $O$ is less than or equal to the $G$-image of
$P$ (that is, intuitively, if ``the two points are less than their
two own images''), \emph{i.e.}, if 
\[
P\sqsubseteq F\left(O\right)\textrm{ and }O\leq G\left(P\right).
\]

Simultaneous post-fixed points are also called \emph{mutually-coinductive
points} of $F$ and $G$.
\begin{itemize}
\item Like for simultaneous pre-fixed points, a point $O\in\mathbb{O}$
or $P\in\mathbb{P}$ can be paired with \emph{more} \emph{than} \emph{one}
point of the other poset to form a single pair of simultaneous post-fixed
points of $F$ and $G$. (Similar to $PreFP_{F,G}$ and $PreFP_{G,F}$,
two functions $PostFP_{F,G}:\mathbb{O}\rightarrow\wp\left(P\right)$
and $PostFP_{G,F}:\mathbb{P}\rightarrow\wp\left(O\right)$ that compute
these sets can be derived from $F$ and $G$.)
\end{itemize}
Two points $O\in\mathbb{O}$, $P\in\mathbb{P}$ are called \emph{simultaneous}
(or \emph{mutual} or \emph{reciprocal}) \emph{fixed points} of $F$
and $G$ if the $F$-image of $O$ is equal to $P$ \emph{and} the
$G$-image of $P$ is equal to $O$, \emph{i.e.}, if 
\[
F\left(O\right)=P\textrm{ and }G\left(P\right)=O.
\]
 (As such, two simultaneous fixed points of $F$ and $G$ are, simultaneously,
simultaneous pre-fixed points of $F$ and $G$ \emph{and} simultaneous
post-fixed points of $F$ and $G$.)
\begin{itemize}
\item Unlike for simultaneous pre-fixed and post-fixed points, a point in
$\mathbb{O}$ or in $\mathbb{P}$ can be paired with \emph{only} \emph{one}
point of the other poset to form a pair of simultaneous fixed points
of $F$ and $G$.\medskip{}
\end{itemize}

\paragraph{Continuity}

Unlike the case for standard (co)induction, where the monotonicity
of generators is enough to guarantee the existence of \emph{least}
and \emph{greatest} standard fixed points, due to the possibility
of multiple pairings in $PreFP$ and $PostFP$, the monotonicity of
two mutual endofunctions $F$ and $G$ seems \emph{not} enough for
proving the existence of least and greatest simultaneous fixed points
of $F$ and $G$.\footnote{It so seemed to us in earlier versions of this paper---but now see
Appendix~\ref{sec:Requiring-Only-Monotonicity}.} However (as can be seen in the proofs in Appendix~\ref{sec:Theorems-and-Proofs}),
the \emph{continuity} of the mutual endofunctions $F$ and $G$ does
guarantee the existence of such fixed points.%
\begin{comment}
TODO: Consider relaxing condition, in future work.
\end{comment}
{} Hence, before proceeding with the formulation of mutual (co)induction
we now introduce the important and useful concept of continuity.
\begin{defn}[Continuous Mutual Endofunctions]
\label{def:Cont}Two mutual endofunctions $F:\mathbb{O}\rightarrow\mathbb{P}$
and $G:\mathbb{P}\rightarrow\mathbb{O}$ defined over two posets $\left(\mathbb{O},\leq\right)$
and $\left(\mathbb{P},\sqsubseteq\right)$ are \emph{continuous} \emph{mutual}
\emph{endofunctions} if and only if for all subsets $\mathbb{M}\subseteq\mathbb{O}$
and $\mathbb{N}\subseteq\mathbb{P}$ whenever greatest lower bounds
(also called \emph{glbs}) $\wedge\mathbb{M}\in\mathbb{O}$ and $\sqcap\mathbb{N}\in\mathbb{P}$
and least upper bounds (also called \emph{lubs}) $\vee\mathbb{M}\in\mathbb{O}$
and $\sqcup\mathbb{N}\in\mathbb{P}$ exist (in $\mathbb{O}$ and in
$\mathbb{P}$) then the corresponding points $\sqcap F\left(\mathbb{M}\right)\in\mathbb{P}$,
$\wedge G\left(\mathbb{N}\right)\in\mathbb{O}$, $\sqcup F\left(\mathbb{M}\right)\in\mathbb{P}$,
and $\vee G\left(\mathbb{N}\right)\in\mathbb{O}$ also exist (in $\mathbb{O}$
and in $\mathbb{P}$), and further, more significantly, we have 
\[
\sqcap F\left(\mathbb{M}\right)=F\left(\wedge\mathbb{M}\right)\textrm{ and }\wedge G\left(\mathbb{N}\right)=G\left(\sqcap\mathbb{N}\right),\textrm{ and}
\]
\[
\sqcup F\left(\mathbb{M}\right)=F\left(\vee\mathbb{M}\right)\textrm{ and }\vee G\left(\mathbb{N}\right)=G\left(\sqcup\mathbb{N}\right).
\]
\end{defn}
As such, if two mutual endofunctions $F:\mathbb{O}\rightarrow\mathbb{P}$
and $G:\mathbb{P}\rightarrow\mathbb{O}$ are continuous, then they
are said to \emph{preserve} the glbs and lubs in $\mathbb{O}$ and
$\mathbb{P}$, whenever these points exist.\footnote{In category theory jargon, if $F$ and $G$ are continuous then $F$
and $G$ are said to \emph{commute} with the meet/glb ($\wedge/\sqcap$)
and join/lub ($\vee/\sqcup$) operations.} (Even though somewhat similar, but the notion of continuity we use
here is stricter and more uniform than the notion of Scott-continuity
used in domain theory, which asserts that mappings preserve only \emph{lubs/joins}
of only \emph{directed} subsets~\cite{Scott82,GunterHandbook90,Abramsky94,Gierz2003,DomTheoryIntro}.)

Figure~\vref{fig:Cont-FG} illustrates \emph{meet-continuous} (also
called \emph{semi-continuous})\emph{ }mutual endofunctions $F$ and
$G$, which preserve only the glbs in $\mathbb{O}$ and $\mathbb{P}$
whenever they exist. Dually, \emph{join-continuous} mutual endofunctions
preserve only the lubs in $\mathbb{O}$ and $\mathbb{P}$ whenever
they exist. As such, mutual endofunctions are (fully) continuous if
and only if they are meet-continuous and join-continuous.

The continuity of two mutual endofunctions has a number of useful
implications, some of which we use to prove that continuity guarantees
the existence of least and greatest simultaneous fixed points.
\begin{itemize}
\item First, if $F$ and $G$ are continuous mutual endofunctions then $F$
and $G$ are also monotonic mutual endofunctions (Lemma~\ref{lem:cont-mono}
in Appendix~\ref{sec:Theorems-and-Proofs}).
\begin{itemize}
\item Note that if $F$ and $G$ are monotonic but \emph{not} continuous
then we can only assert that the points $F\left(\wedge\mathbb{M}\right)$
and $G\left(\sqcap\mathbb{N}\right)$ are lower bounds of the sets
$F\left(\mathbb{M}\right)$ and $G\left(\mathbb{N}\right)$, but not
that they are necessarily the \emph{greatest} lower bounds of these
sets, and we can also only assert that the elements $F\left(\vee\mathbb{M}\right)$
and $G\left(\sqcup\mathbb{N}\right)$ are upper bounds of the two
sets, respectively, but not that they are necessarily the \emph{least}
upper bounds of these sets. More precisely, if $F$ and $G$ are monotonic,
but not necessarily continuous, then we only have 
\[
F\left(\wedge\mathbb{M}\right)\sqsubseteq\sqcap F\left(\mathbb{M}\right)\textrm{ and }G\left(\sqcap\mathbb{N}\right)\leq\wedge G\left(\mathbb{N}\right),\textrm{ and}
\]
\[
\sqcup F\left(\mathbb{M}\right)\sqsubseteq F\left(\vee\mathbb{M}\right)\textrm{ and }\vee G\left(\mathbb{N}\right)\leq G\left(\sqcup\mathbb{N}\right).
\]
Compared to monotonicity, as such, the continuity of $F$ and $G$
can be seen as requiring or asserting the \emph{equality} of these
points whenever they exist, intuitively thereby ``allowing no elements
in between''. Continuity, thus, is said to allow functions $F$ and
$G$ that ``have no jumps'' or ``have no surprises''.%
\begin{comment}
Now unneeded! (See~\cite{AbdelGawad2019c} for intuitions on mutual
continuity.)
\end{comment}
\end{itemize}
\item Second, if $F$ and $G$ are continuous mutual endofunctions, then
the compositions $G\circ F$ and $F\circ G$ are (standard) continuous
endofunctions, and they are monotonic endofunctions too (Lemma~\ref{lem:comp-mono}
in Appendix~\ref{sec:Theorems-and-Proofs}).
\begin{itemize}
\item Further, if $\mathbb{O}$ and $\mathbb{P}$ are complete lattices,
then the composition functions $G\circ F:\mathbb{O}\rightarrow\mathbb{O}$
and $F\circ G:\mathbb{P}\rightarrow\mathbb{P}$ (like any standard
monotonic endofunctions over complete lattices) have standard pre-fixed
points and standard post-fixed points in $\mathbb{O}$ and $\mathbb{P}$
respectively. Even further, the components $O\in\mathbb{O}$ and $P\in\mathbb{P}$
of simultaneous fixed points of $F$ and $G$ are always, each component
individually, among the standard fixed points of the compositions
$G\circ F$ and $F\circ G$ (Lemma~\ref{lem:comp-fp} in Appendix~\ref{sec:Theorems-and-Proofs}).
The converse, however, does not necessarily hold.
\end{itemize}
\item Third, if $F$ and $G$ are continuous mutual endofunctions then $F$
and $G$ (and their compositions) map complete lattices to complete
lattices (note that, by definition, the empty set is not a complete
lattice, and that a singleton poset is a trivial complete lattice).
In other words, if subsets $\mathbb{M}\subseteq\mathbb{O}$ and $\mathbb{N}\subseteq\mathbb{P}$
happen to be complete sublattices of $\mathbb{O}$ and $\mathbb{P}$
then $F\left(\mathbb{M}\right)$ and $G\left(\mathbb{N}\right)$ are
either empty or are complete \emph{sublattices} of $\mathbb{P}$ and
$\mathbb{O}$ respectively (Lemma~\ref{lem:cont-complatt} in Appendix~\ref{sec:Theorems-and-Proofs}).
\item Fourth and finally, if $F$ and $G$ are continuous mutual endofunctions,
then for all $O\in\mathbb{O}$ and $P\in\mathbb{P}$ the posets $PreFP_{F,G}\left(O\right)$
and $PreFP_{G,F}\left(P\right)$ are complete lattices (Lemma~\ref{lem:PreFP-img-complatt}
in Appendix~\ref{sec:Theorems-and-Proofs}). Continuity also guarantees%
\begin{comment}
TODO: True?
\end{comment}
{} the existence of least simultaneous pre-fixed points (Lemma~\ref{lem:PreFP-complatt}
in Appendix~\ref{sec:Theorems-and-Proofs}). We use continuity to
prove that the least simultaneous pre-fixed points are also the least
simultaneous \emph{fixed} points (\nameref{thm:SFP} in Appendix~\ref{sec:Theorems-and-Proofs}).\footnote{Note that singleton posets are trivial complete lattices. As such,
this condition is already satisfied in an encoding of standard (co)induction
using mutual (co)induction. That is because in such an encoding a
given pre-fixed point is an element that is paired with precisely
one element, namely itself, to form a pair of simultaneous pre-fixed
points. Thus, in mutual (co)induction encodings of standard (co)induction,
the sets $PreFP_{F,G}\left(O\right)$, $PreFP_{G,F}\left(P\right)$,
$PostFP_{F,G}\left(O\right)$, and $PostFP_{G,F}\left(P\right)$ are
always either $\phi$ or singleton posets for all $O\in\mathbb{O}$
and all $P\in\mathbb{P}$.}
\end{itemize}
\medskip{}

Having made a digression to discuss the continuity of mutual endofunctions,
and some of its implications, we now resume our formulation of mutual
(co)induction.

Now, if $\leq$ is a complete lattice over $\mathbb{O}$ and $\sqsubseteq$
is a complete lattice over $\mathbb{P}$ (\emph{i.e.}, if $\leq$
is an ordering relation where meets $\wedge$ and joins $\vee$ of
\emph{all} subsets of $\mathbb{O}$ are guaranteed to exist in $\mathbb{O}$,
and similarly for $\sqsubseteq$, $\sqcap$, and $\sqcup$ in $\mathbb{P}$)
and if $F$ and $G$ are\emph{ continuous} mutual endofunctions over
$\mathbb{O}$ and $\mathbb{P}$ then we have the following:
\begin{itemize}
\item $F$ and $G$ are called \emph{simultaneous} \emph{generating} \emph{functions}
(or \emph{simultaneous generator}s or \emph{mutual generators }or
\emph{reciprocal generators}),
\item the \emph{least simultaneous pre-fixed points }of $F$ and $G$, called
$\mu_{F}$ and $\mu_{G}$, exist in $\mathbb{O}$ and $\mathbb{P}$,
\item together, the points $\mu_{F}$ and $\mu_{G}$ are also the \emph{least
simultaneous fixed points} of $F$ and $G$ (as we prove in~\nameref{thm:SFP}),
\item the \emph{greatest simultaneous post-fixed points} of $F$ and $G$,
called $\nu_{F}$ and $\nu_{G}$, exist in $\mathbb{O}$ and $\mathbb{P}$,
and 
\item together, the points $\nu_{F}$ and $\nu_{G}$ are also the \emph{greatest
simultaneous fixed points} of $F$ and $G$ (as we prove in~\nameref{thm:SFP}).
\end{itemize}
Further, given that $\mu_{F}$ and $\mu_{G}$ are the \emph{least}
simultaneous pre-fixed points\emph{ }of $F$ and $G$ and $\nu_{F}$
and $\nu_{G}$ are the \emph{greatest }simultaneous post-fixed points\emph{
}of $F$ and $G$, for any element $O\in\mathbb{O}$ and $P\in\mathbb{P}$
we have:
\begin{itemize}
\item (\emph{mutual induction})\phantom{co} if $F\left(O\right)\sqsubseteq P$
and $G\left(P\right)\leq O$, then $\mu_{F}\leq O$ and $\mu_{G}\sqsubseteq P$,\medskip{}
\\
which, in words, means that if $O$ and $P$ are simultaneous pre-fixed/inductive/large
points of $F$ and $G$, then points $\mu_{F}$ and $\mu_{G}$ are
less than or equal to $O$ and $P$ (\emph{i.e.}, $\mu_{F}$ and $\mu_{G}$
are the smallest simultaneously-large points of $F$ and $G$), and,
\item (\emph{mutual coinduction}) if $P\sqsubseteq F\left(O\right)$ and
$O\leq G\left(P\right)$, then $O\leq\nu_{F}$ and $P\sqsubseteq\nu_{G}$,\medskip{}
\\
which, in words, means that if $O$ and $P$ are simultaneous post-fixed/coinductive/small
points of $F$ and $G$, then points $O$ and $P$ are less than or
equal to points $\nu_{F}$ and $\nu_{G}$ (\emph{i.e.}, $\nu_{F}$
and $\nu_{G}$ are the largest simultaneously-small points of $F$
and $G$).
\end{itemize}

\section{\label{sec:Related-Work}Related Work}

The work closest to the one we present here seems to be that of Paulson,
presented \emph{e.g.} in~\cite{Paulson1998}, to support the development
of the Isabelle proof assistant~\cite{Isabelle2015}. We already
mentioned in $\mathsection$\ref{sec:Introduction} the influence
of Paulson's work on motivating our definition of mutual coinduction
(we discuss this motivation, and others, in more detail in~\cite{AbdelGawad2019e}).
Due to Paulson's interest in making some form of coinduction available
in systems such as Isabelle~\cite{Paulson2019}, Paulson was interested
only in the set-theoretic definition of standard (co)induction and
mutual (co)induction~\cite{Paulson1995} (the set-theoretic definition,
according to~\cite{Paulson2019}, `was definitely the easiest to
develop, especially during the 1990s, when no general mechanisation
of lattice theory was even available').

More technically, it seems Paulson was interested in requiring generators
to be monotonic (as opposed to requiring their continuity, which is
sometimes viewed as an undesirably strong assumption~\cite[p.72]{Priestley2002}).
As such, Paulson used monotonic generators over \emph{the powerset
of a disjoint sum} domain so as to define (or, rather, encode) mutual
set-theoretic (co)induction using standard set-theoretic (co)induction\footnote{Via noting an order-isomorphism between $\wp\left(\sum_{i}D_{i}\right)$
(the powerset of a disjoint sum) and $\prod_{i}\wp\left(D_{i}\right)$
(the product of powersets). In our opinion the use of the disjoint
sum in Paulson's definition of mutual (co)induction, while technically
clever, is unnatural and unintuitive (as demonstrated, \emph{e.g.},
in the examples of~\cite[$\mathsection$4.5]{Paulson1995}).}. Additionally, in~\cite[p.33]{Paulson1995} Paulson stated that
the standard Fixed Point Theorem has been proven in Isabelle `only
for a simple powerset lattice,' which made Paulson limit his interests
to such ``simple powerset lattices,'' even when the theorem applies
to \emph{any} complete lattice~\cite{Davey2002,Rol2002}, not just
to a particular instance (\emph{i.e.}, not only to powerset lattices).
As such, in summary, it seems to us that Paulson, for considerations
related to his interests in automated theorem proving, was not interested
in considering the continuity of generators in his work.

While not particularly aimed at semantics, Paulson's work on mutual
(co)induction, indirectly, provided semantics for mutual (co)inductive
datatypes in \emph{functional} programming languages (\emph{e.g.}
ML), where mutual datatype constructors were modeled by mutual generators.
Functional programming languages, however, are largely \emph{struc\-turally}-typed
and structurally-sub\-typed. Given that we have a general interest,
rather, in providing semantics for datatypes in mainstream \emph{object-oriented}
programming languages (such as Java, C\#, Scala, and Kotlin), which
are typically \emph{nomin\-ally}-typed and nominally-sub\-typed
programming languages, our interest is more in the order-theoretic
formulation of mutual (co)induction.\footnote{See~\cite{AbdelGawad2019f,AbdelGawad2019h}, \cite{AbdelGawad2019d}
and also~\cite{AbdelGawad2019a,AbdelGawad2019b,AbdelGawad2019e}
for a discussion of why order-theory and category-theory seem to be
more suited than set-theory for modeling nominally-typed and nominally-subtyped
programming languages.}

Given the difference between our work and Paulson's work regarding
how mutual (co)induction is technically formulated, motivated by the
different goals behind the two formulations, we anticipate that results
in both works (particularly regarding the definition of simultaneous
fixed points) are not in one-to-one correspondence with each other.
We, however, keep a more detailed comparison of the formulation of
mutual coinduction we present here to the formulation of Paulson---which
may reveal more similarities and resolve some of the differences between
the two formalizations---for future work.

Another work that is related to ours is the work presented in~\cite{Blanchette2014}.
Since the work in~\cite{Blanchette2014} builds on that of Paulson,
and has similar aims in supporting the development of Isabelle, the
work in~\cite{Blanchette2014} adopts the same specific set-theoretic
view of mutual (co)induction as that of Paulson.

\begin{comment}
TODO: Mention work related to mutual induction, to coinduction, and
standard induction (each separately).
\end{comment}

\begin{comment}
Other less-relevant related work: \cite{Greiner92,Kaser1993,Yu1997,Brandt1998,Hughes2001,Rol2002,Geldrop2009,Leino2013,Kozen2016,Fong2018}
\end{comment}

\begin{comment}
TODO: OOP coalgebra (pre-generics, value level only? no type level?)\cite{Jacobs95,Poll2000276,Jacobs2002}
\end{comment}

\section{\label{sec:Conclusion-and-FW}Conclusion and Future Work}

Standard induction (which includes the standard notions of \emph{mathematical}
\emph{induction} and \emph{structural induction}) is well-known, and
it is relatively easy reason about. Standard \emph{coinduction} is
also known, but it is a bit shrouded in mystery and unreasonability\footnote{Likely due to its strong connections with negation~\cite{Kozen2016,AbdelGawad2019a}.}.
Mutual induction is also known, if somewhat to a lesser extent than
standard induction. Mutual induction is a bit harder to reason about
than standard induction however. Mutual \emph{co}induction---our
main interest in this paper---is, however, almost unknown, and has
(so far) been perceived as being both mysterious and hard to reason
about. We hope that this paper, via presenting the definition of mutual
coinduction as a simple generalization of the order-theoretic definition
of standard coinduction, has put mutual coinduction into more familiar
light, and that, by presenting a proof of a related proof principle,
it has also made mutual coinduction simpler to reason about.%
\begin{comment}
It is our hope that this paper and other related articles~\cite{AbdelGawad2019c,AbdelGawad2019d,AbdelGawad2019e}
have made mutual coinduction more familiar, less mysterious, and also
a bit simpler to reason about.
\end{comment}

While the continuity condition on generators in our formulation of
mutual (co)in\-duction is sufficient for proving the existence of
least and greatest simultaneous fixed points in complete lattices
(while monotonicity seems insufficient), yet it is not clear to us
whether (full) continuity is \emph{necessary} for such a proof. It
may be useful to consider, in some future work, the possibility of
relaxing the continuity condition, while still guaranteeing the existence
of simultaneous fixed points. In particular, it may be useful to consider
the effect of having other more liberal continuity conditions, such
as Scott-continuity, on the existence of simultaneous fixed points.
It may be also useful to study simultaneous pre-fixed points and simultaneous
post-fixed points that are \emph{not} necessarily fixed points (as
is done, for example, in the study of algebras and coalgebras in category
theory).

As another possible future work that can build on the definition of
mutual coinduction we present here, it may be useful to consider defining
\emph{infinite mutual coinduction}, which, as we conceive it, generalizes
mutual coinduction to involve an infinite (countable, or even uncountable!)
number of orderings and generators. As of the time of this writing,
we are not aware of immediate applications of infinite mutual coinduction.
Given the mystery surrounding both coinduction and some particular
areas of science, though, we conjecture that infinite mutual coinduction
(if it is indeed reasonably definable) may have applications in areas
of science such as quantum physics, \emph{e.g.}, by it offering mathematical
models of quantum phenomena such as superposition, entanglement, and/or
interference\footnote{For example, Penrose calls in~\cite{Penrose2004} for some new kind
of mathematics for having an accurate understanding of our universe.
Perhaps infinite mutual coinduction is a piece of such ``new math.''}. In agreement with this conjecture, we also intuit and conjecture
that infinite mutual coinduction may have an impact on quantum computing,
including reasoning about quantum programs and quantum software.

\bibliographystyle{plain}

\appendix

\section{\label{sec:Theorems-and-Proofs}Lemmas, Theorems and Proofs}

In this appendix we present proofs for the lemmas and theorems of~$\mathsection$\ref{sec:Mutual-(Co)Induction}.

\subsection{Supporting Lemmas}
\begin{lem}[Continuous Functions are Monotonous]
\label{lem:cont-mono}If $F:\mathbb{O}\rightarrow\mathbb{P}$ and
$G:\mathbb{P}\rightarrow\mathbb{O}$ are continuous mutual endofunctions
over posets $\left(\mathbb{O},\leq\right)$ and $\left(\mathbb{P},\sqsubseteq\right)$
then $F$ and $G$ are also monotonic, \emph{i.e.}, for all $O_{1},O_{2}\in\mathbb{O}$
\[
O_{1}\leq O_{2}\implies F\left(O_{1}\right)\sqsubseteq F\left(O_{2}\right)
\]
and for all $P_{1},P_{2}\in\mathbb{P}$
\[
P_{1}\sqsubseteq P_{2}\implies G\left(P_{1}\right)\leq G\left(P_{2}\right).
\]
\end{lem}
\begin{proof}
Let $O_{1},O_{2}\in\mathbb{O}$ such that $O_{1}\leq O_{2}$ (\emph{i.e.},
in Definition~\ref{def:Cont}, take $\mathbb{M}=\left\{ O_{1},O_{2}\right\} $
where $O_{1}\leq O_{2}$). From the definition of $\sqcap$ (as a
greatest \emph{lower} bound), we particularly have 
\begin{equation}
\bigsqcap F\left(\left\{ O_{1},O_{2}\right\} \right)=\bigsqcap\left\{ F\left(O_{1}\right),F\left(O_{2}\right)\right\} \sqsubseteq F\left(O_{2}\right).\label{eq:cont-mono}
\end{equation}
 By the continuity of $F$, we also have $\bigsqcap F\left(\left\{ O_{1},O_{2}\right\} \right)=F\left(\bigwedge\left\{ O_{1},O_{2}\right\} \right)$.
Given that $O_{1}\leq O_{2}$, we have $\bigwedge\left\{ O_{1},O_{2}\right\} =O_{1}\wedge O_{2}=O_{1}$
and thus, further, we have $F\left(\bigwedge\left\{ O_{1},O_{2}\right\} \right)=F\left(O_{1}\right)$.
As such, we have $\bigsqcap F\left(\left\{ O_{1},O_{2}\right\} \right)=F\left(O_{1}\right)$.
Substituting for $\bigsqcap F\left(\left\{ O_{1},O_{2}\right\} \right)$
(the l.h.s.) in~(\ref{eq:cont-mono}), we get $F\left(O_{1}\right)\sqsubseteq F\left(O_{2}\right)$,
as required.

Similarly for $G$.
\end{proof}
\begin{lem}[Composition Preserves Monotonicity and Continuity]
\label{lem:comp-mono}If $F:\mathbb{O}\rightarrow\mathbb{P}$ and
$G:\mathbb{P}\rightarrow\mathbb{O}$ are monotonic (continuous) mutual
endofunctions over posets $\mathbb{O}$ and $P$, then the compositions
$G\circ F:\mathbb{O}\rightarrow\mathbb{O}$ and $F\circ G:\mathbb{P}\rightarrow\mathbb{P}$
are monotonic (continuous) endofunctions over $\mathbb{O}$ and $\mathbb{P}$
respectively.
\end{lem}
\begin{proof}
By substitution, since $F\left(O_{1}\right)\in\mathbb{P}$ and $F\left(O_{2}\right)\in\mathbb{P}$
for all $O_{1},O_{2}\in\mathbb{O}$, then, by the monotonicity of
$F$ then that of $G$, we have 
\[
O_{1}\leq O_{2}\implies F\left(O_{1}\right)\sqsubseteq F\left(O_{2}\right)\implies G\left(F\left(O_{1}\right)\right)\leq G\left(F\left(O_{2}\right)\right),
\]
\emph{i.e.}, that $G\circ F$ is monotonic.

Similarly, for all $P_{1},P_{2}\in\mathbb{P}$, the monotonicity of
$G$ then of $F$ implies
\[
P_{1}\sqsubseteq P_{2}\implies G\left(P_{1}\right)\leq G\left(P_{2}\right)\implies F\left(G\left(P_{1}\right)\right)\sqsubseteq F\left(G\left(P_{2}\right)\right),
\]
\emph{i.e.}, that $F\circ G$ is monotonic.

Using a very similar but much more tedious %
\begin{comment}
substitution 
\end{comment}
argument we can prove that composition preserves continuity.
\end{proof}
\begin{lem}[Components of Simultaneous Fixed Points are Standard Fixed Points
of Compositions]
\noindent \label{lem:comp-fp}If $F:\mathbb{O}\rightarrow\mathbb{P}$
and $G:\mathbb{P}\rightarrow\mathbb{O}$ are monotonic mutual endofunctions
over posets $\mathbb{O}$ and $\mathbb{P}$, then the components of
their simultaneous pre-/post-/fixed points are standard pre-/post-/fixed
points of the compositions $G\circ F:\mathbb{O}\rightarrow\mathbb{O}$
and $F\circ G:\mathbb{P}\rightarrow\mathbb{P}$.
\end{lem}
\begin{proof}
For each pair of points $O\in\mathbb{O}$ and $P\in\mathbb{P}$ of
simultaneous pre-fixed points of $F$ and $G$, by the definition
of simultaneous pre-fixed points we have both inequalities 
\[
F\left(O\right)\sqsubseteq P\textrm{ and }G\left(P\right)\leq O.
\]
 Applying $G$ to both sides of the first inequality, the monotonicity
of $G$ implies that $G\left(F\left(O\right)\right)\leq G\left(P\right)$.\footnote{Note that continuity is not needed here.}
Combining this with the second inequality via the common expression
$G\left(P\right)$, we have
\[
G\left(F\left(O\right)\right)\leq G\left(P\right)\leq O.
\]
Then, by the transitivity of $\leq$,
\[
G\left(F\left(O\right)\right)\leq O
\]
\emph{i.e.}, point $O$ is a standard pre-fixed point of the composition
$G\circ F$.

Symmetrically, by applying $F$ to both sides of the second inequality,
the monotonicity of $F$ implies that
\[
F\left(G\left(P\right)\right)\sqsubseteq F\left(O\right)\sqsubseteq P.
\]
and thus, by the transitivity of $\sqsubseteq$, point $P$ is a standard
pre-fixed point of the composition $F\circ G$. (See Figure~\ref{fig:sim-prefix-latt}
for illustration.)

A dual argument implies that components of simultaneous post-fixed
points of $F$ and $G$ are standard post-fixed points of $G\circ F$
and $F\circ G$ respectively.

Combining both results implies that components of simultaneous fixed
points of $F$ and $G$ are also standard fixed points of $G\circ F$
and $F\circ G$.
\end{proof}
\begin{lem}[Continuous Functions Preserve Complete Lattices]
\label{lem:cont-complatt}If $F:\mathbb{O}\rightarrow\mathbb{P}$
and $G:\mathbb{P}\rightarrow\mathbb{O}$ are continuous mutual endofunctions
over posets $\mathbb{O}$ and $\mathbb{P}$, then for all subsets
$\mathbb{M}\subseteq\mathbb{O}$ and $\mathbb{N}\subseteq\mathbb{P}$,
if $\mathbb{M}$ and $\mathbb{N}$ are complete sublattices of $\mathbb{O}$
and $\mathbb{P}$ then the images $F\left(\mathbb{M}\right)$ and
$G\left(\mathbb{N}\right)$ are also complete sublattices of $\mathbb{P}$
and $\mathbb{O}$, respectively.
\end{lem}
\begin{proof}
First, let's consider the case where $\mathbb{M}\subseteq\mathbb{O}$
is a complete lattice. By the definition of a complete lattice, the
points $\wedge\mathbb{A}$ and $\vee\mathbb{A}$ exist in $\mathbb{M}$
for all subsets $\mathbb{A}\subseteq\mathbb{M}$. As such, the points
$F\left(\wedge\mathbb{A}\right)$ and $F\left(\vee\mathbb{A}\right)$
exist in $\mathbb{P}$, and, accordingly, by the definition of $F\left(\mathbb{M}\right)$,
are also members of $F\left(\mathbb{M}\right)$. By the continuity
of $F$, we also have $F\left(\wedge\mathbb{A}\right)=\sqcap F\left(\mathbb{A}\right)$
and $F\left(\vee\mathbb{A}\right)=\sqcup F\left(\mathbb{A}\right)$.
As such, for any image set $F\left(\mathbb{A}\right)\subseteq\mathbb{P}$
a greatest lower bound, $\sqcap F\left(\mathbb{A}\right)$, and a
least upper bound, $\sqcup F\left(\mathbb{A}\right)$, exist. Further,
because of the continuity of $F$, these two points are members of
$F\left(\mathbb{M}\right)$. Thus, all subsets of $\mathbb{M}$ have
images of their glbs and lubs in $F\left(\mathbb{M}\right)$. That
by itself does not, however, prove that the set $F\left(\mathbb{M}\right)$
is a complete lattice, yet.

To prove that $F\left(\mathbb{M}\right)$ is a complete lattice, we
have to prove that the points $\sqcap\mathbb{B}$ and $\sqcup\mathbb{B}$
exist in $F\left(\mathbb{M}\right)$ for all sets $\mathbb{B}\subseteq F\left(\mathbb{M}\right)\subseteq\mathbb{P}$.
Given that $\mathbb{B}$ is subset of $F\left(\mathbb{M}\right)$,
the image of $\mathbb{M}$, then there exists some (one or more) set
$\mathbb{A}\subseteq\mathbb{M}$ such that $F\left(\mathbb{A}\right)=\mathbb{B}$.
Pick one such set $\mathbb{A}$.\footnote{Sounds like the Axiom of Choice is needed here.}
Then, for that particular $\mathbb{A}$, we have $\sqcap\mathbb{B}=\sqcap F\left(\mathbb{A}\right)$
and $\sqcup\mathbb{B}=\sqcup F\left(\mathbb{A}\right)$. Since we
proved that for all sets $\mathbb{A}\subseteq\mathbb{M}$ the points
$\sqcap F\left(\mathbb{A}\right)$ and $\sqcup F\left(\mathbb{A}\right)$
are members of $F\left(\mathbb{M}\right)$, we conclude that $\sqcap\mathbb{B}$
and $\sqcup\mathbb{B}$ are members of $F\left(\mathbb{M}\right)$
for all sets $\mathbb{B}\subseteq F\left(\mathbb{M}\right)$. As such,
the set $F\left(\mathbb{M}\right)$ is a complete lattice.

Next, by a symmetric argument, if $\mathbb{N}\subseteq\mathbb{P}$
is a complete lattice then, by the continuity of $G$, the set $G\left(\mathbb{N}\right)$
is also a complete lattice, as required.
\end{proof}
\begin{lem}[Component Images of Simultaneous Pre-/Post-Fixed Points form Complete
Lattices]
\label{lem:PreFP-img-complatt}If $F:\mathbb{O}\rightarrow\mathbb{P}$
and $G:\mathbb{P}\rightarrow\mathbb{O}$ are continuous mutual endofunctions
over complete lattices $\mathbb{O}$ and $\mathbb{P}$, then, for
all $O\in\mathbb{O}$ and $P\in\mathbb{P}$, the component image sets
$PreFP_{F,G}\left(O\right)$, $PreFP_{G,F}\left(P\right)$, $PostFP_{F,G}\left(O\right)$,
and $PostFP_{G,F}\left(P\right)$, as defined in~$\mathsection$\ref{sec:Mutual-(Co)Induction}\emph{,}
are either empty or are complete lattices.
\end{lem}
\begin{proof}
For a point $O\in\mathbb{O}$ such that $PreFP_{F,G}\left(O\right)$
is nonempty, define 
\[
\mathbb{P}_{O}=PreFP_{F,G}\left(O\right)=\left\{ P\in\mathbb{P}|F\left(O\right)\sqsubseteq P\textrm{ and }G\left(P\right)\leq O\right\} .
\]
Since $\mathbb{P}$ is a complete lattice, the set $\mathbb{P}_{O}$
(as a subset of $\mathbb{P}$) has a greatest lower bound $\sqcap\mathbb{P}_{O}$
and a least upper bound $\sqcup\mathbb{P}_{O}$ that are members of
$\mathbb{P}$. To prove that $\mathbb{P}_{O}$ itself is a complete
lattice, first we prove that these two points (\emph{i.e.}, the glb
and lub of $\mathbb{P}_{O}$) are members of $\mathbb{P}_{O}$.

Let's note that for all points $O$ where $\mathbb{P}_{O}$ is non-empty
we always have the point $F\left(O\right)$ as a member of $\mathbb{P}_{O}$.
That's because, by the reflexivity of $\sqsubseteq$, we have $F\left(O\right)\sqsubseteq F\left(O\right)$.
Further, using Lemma~\ref{lem:comp-fp}, we have $G\left(F\left(O\right)\right)\leq O$.
As such, by the definition of simultaneous pre-fixed points, the points
$O\in\mathbb{O}$ and $F\left(O\right)\in\mathbb{P}$ are simultaneous
pre-fixed points of $F$ and $G$.\footnote{This brings one of the most delicate points in proving the mutual
(co)induction principles. Note that $O$ and $F\left(O\right)$ are
\emph{not} necessarily simultaneous pre-fixed points for \emph{all}
$O\in\mathbb{O}$, but that, according to Lemma~\ref{lem:comp-fp},
points $O$ and $F\left(O\right)$ are simultaneous pre-fixed points
\emph{only} whenever there is some $P\in\mathbb{P}$ (possibly equal
to $F\left(O\right)$, and possibly not) such that $O$ and $P$ are
simultaneous pre-fixed points of $F$ and $G$, \emph{i.e.}, such
that $P$ witnesses, via $G\left(P\right)$, that $G\left(F\left(O\right)\right)\leq O$.
In particular, it is \emph{not} necessarily true that for all $O\in\mathbb{O}$
we have $G\left(F\left(O\right)\right)\leq O$ (Otherwise, \emph{all}
elements $O\in\mathbb{O}$ would have formed simultaneous pre-fixed
points of $F$ and $G$, simply by pairing each $O$ with $F\left(O\right)$.
This goes counter to intuitions about mutual (co)induction, since
it would eventually lead to concluding that all points of $\mathbb{O}$---and
similarly of $\mathbb{P}$---are simultaneous \emph{fixed} points
of $F$ and $G$!).

Readers should be aware of this delicate and tricky point specific
to mutual (co)induction. That's because this point \emph{has} \emph{no}
\emph{counterpart} in standard (co)induction (or, at least, has no
obvious counterpart, since in an encoding of standard (co)induction
using mutual (co)induction, where say $G=\mathbf{1}$, we will have
$G\left(F\left(O\right)\right)=F\left(O\right)\leq O$ only if $F\left(O\right)\leq O$).} Hence, $F\left(O\right)\in\mathbb{P}_{O}$. Given that, by the definition
of $\mathbb{P}_{O}$, the point $F\left(O\right)$ is less than or
equal to all members of $\mathbb{P}_{O}$, we have 
\[
\sqcap\mathbb{P}_{O}=F\left(O\right),
\]
and, as such, the greatest lower bound of $\mathbb{P}_{O}$ is a member
of $\mathbb{P}_{O}$, \emph{i.e.}, we have $\sqcap\mathbb{P}_{O}\in\mathbb{P}_{O}$
as needed.

For $\sqcup\mathbb{P}_{O}$, we prove that it is a member of $\mathbb{P}_{O}$
more directly, using the definition of $\mathbb{P}_{O}$ and the continuity
of $G$. As for any member of $\mathbb{P}$, for $\sqcup\mathbb{P}_{O}$
to be a member of $\mathbb{P}_{O}$ we must have $F\left(O\right)\sqsubseteq\sqcup\mathbb{P}_{O}$
and $G\left(\sqcup\mathbb{P}_{O}\right)\leq O$. The first condition
is satisfied since we just proved that $F\left(O\right)$ is exactly
$\sqcap\mathbb{P}_{O}$, and, as for any set, we have $\sqcap\mathbb{P}_{O}\sqsubseteq\sqcup\mathbb{P}_{O}$
whenever such points exist, and as such we have $F\left(O\right)\sqsubseteq\sqcup\mathbb{P}_{O}$.
For the second condition (this is where continuity is needed), we
have $G\left(\sqcup\mathbb{P}_{O}\right)=\vee G\left(\mathbb{P}_{O}\right)$
by continuity. Since, by the definition of $\mathbb{P}_{O}$, all
members of $G\left(\mathbb{P}_{O}\right)$ are less than or equal
to $O$, then $O$ is an upper bound of $G\left(\mathbb{P}_{O}\right)$.
As such, we have $\vee G\left(\mathbb{P}_{O}\right)\leq O$. Hence,
for $\sqcup\mathbb{P}_{O}$, we have $G\left(\sqcup\mathbb{P}_{O}\right)=\vee G\left(\mathbb{P}_{O}\right)\leq O$,
as required. Hence, $\sqcup\mathbb{P}_{O}\in\mathbb{P}_{O}$ as needed.

Since the greatest lower bound of $\mathbb{P}_{O}$ and the least
upper bound of $\mathbb{P}_{O}$ are members of $\mathbb{P}_{O}$,
so far we can assert that the set $\mathbb{P}_{O}=PreFP_{F,G}\left(O\right)$
is a bounded poset.

To prove that $\mathbb{P}_{O}$ is not only a bounded poset but, rather,
that it is a complete lattice, we have to also consider \emph{proper}
subsets of $\mathbb{P}_{O}$. The argument for proper subsets of $\mathbb{P}_{O}$
is very similar to the one we just used for $\mathbb{P}_{O}$ (as
an improper subset of itself). In particular, let $\mathbb{N}\subset\mathbb{P}_{O}$
be some proper subset of $\mathbb{P}_{O}$ (\emph{i.e.}, is some set
of points of $\mathbb{P}$ that, paired with $O$, are simultaneous
pre-fixed points of $F$ and $G$). Again, since $\mathbb{P}$ is
a complete lattice, the elements $\sqcap\mathbb{N}$ and $\sqcup\mathbb{N}$
exist. We proceed to prove that these points are also members of $\mathbb{P}_{O}$.

In particular, to be members of $\mathbb{P}_{O}$ the two points $\sqcap\mathbb{N}$
and $\sqcup\mathbb{N}$ have to satisfy the membership condition of
$\mathbb{P}_{O}$, \emph{i.e.}, they have to form, when paired with
$O$, simultaneous pre-fixed points of $F$ and $G$. Again using
the continuity of $G$, we can see that this is true since, like we
had for $\sqcap\mathbb{P}_{O}$ and $\sqcup\mathbb{P}_{O}$, we have
\[
F\left(O\right)\sqsubseteq\sqcap\mathbb{N}\textrm{ and }G\left(\sqcap\mathbb{N}\right)=\wedge G\left(\mathbb{N}\right)\leq O,\textrm{ and}
\]
\[
F\left(O\right)\sqsubseteq\sqcup\mathbb{N}\textrm{ and }G\left(\sqcup\mathbb{N}\right)=\vee G\left(\mathbb{N}\right)\leq O.
\]

As such, points $\sqcap\mathbb{N}$ and $\sqcup\mathbb{N}$ are members
of $\mathbb{P}_{O}$. Thus, $\mathbb{P}_{O}$ is a complete lattice.

Using a symmetric argument and the continuity of $F$, we can also
prove that the set $\mathbb{O}_{P}=PreFP_{G,F}\left(P\right)$ is
either empty or is a complete lattice for all $P\in\mathbb{P}$. (Figure~\vref{fig:sim-prefix-latt}
illustrates the sets $\mathbb{P}_{O}$ and $\mathbb{O}_{P}$.)
\begin{figure}
\noindent \begin{centering}
\includegraphics[scale=0.5]{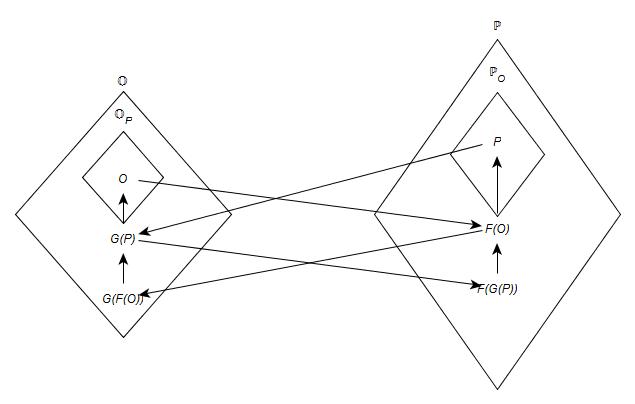}
\par\end{centering}
\caption{\label{fig:sim-prefix-latt}Illustrating Lemma~\ref{lem:comp-fp}
and Lemma~\ref{lem:PreFP-img-complatt}.}
\end{figure}

Dually, we can also prove that, for all $O\in\mathbb{O}$ and $P\in\mathbb{P}$,
the sets $PostFP_{F,G}\left(O\right)$ and $PostFP_{G,F}\left(P\right)$
are either empty or are complete lattices, as required.
\end{proof}
\begin{lem}[Components of Pre-/Post-Fixed Points form Complete Lattices]
\label{lem:PreFP-complatt}If $F:\mathbb{O}\rightarrow\mathbb{P}$
and $G:\mathbb{P}\rightarrow\mathbb{O}$ are continuous mutual endofunctions
over complete lattices $\mathbb{O}$ and $\mathbb{P}$, then the sets
\begin{equation}
\mathbb{C}=\left\{ O\in\mathbb{O}|\exists P\in\mathbb{P}.F\left(O\right)\sqsubseteq P\textrm{ and }G\left(P\right)\leq O\right\} ,\textrm{ and}\label{eq:allPreO}
\end{equation}
\begin{equation}
\mathbb{D}=\left\{ P\in\mathbb{P}|\exists O\in\mathbb{O}.F\left(O\right)\sqsubseteq P\textrm{ and }G\left(P\right)\leq O\right\} \label{eq:allPreP}
\end{equation}
of all components of simultaneous pre-fixed points are complete sublattices
of $\mathbb{O}$ and $\mathbb{P}$, respectively. Similarly for simultaneous
post-fixed points.
\end{lem}
\begin{proof}
First, let's note that the definitions of $\mathbb{C}$ and $\mathbb{D}$
mean that $\mathbb{C}$ is the set of all $O\in\mathbb{O}$ where
there is some $P\in\mathbb{P}$ such that the $F$-image of $O$ is
less than $P$ and the $G$-image of $P$ is less than $O$ and, symmetrically,
that $\mathbb{D}$ is the set of all $P\in\mathbb{P}$ where there
is some $O\in\mathbb{O}$ such that the $F$-image of $O$ is less
than $P$ and the $G$-image of $P$ is less than $O$. (As such,
the variables $O$ and $P$ in Equations~(\ref{eq:allPreO}) and~(\ref{eq:allPreP})
range over the set of all simultaneous pre-fixed points of $F\textrm{ and }G$.)

Note also that $\mathbb{C}=\bigcup_{P\in\mathbb{P}}PreFP_{G,F}\left(P\right)$
and $\mathbb{D}=\bigcup_{O\in\mathbb{O}}PreFP_{F,G}\left(O\right)$,
but that Lemma~\ref{lem:PreFP-img-complatt} does not (by itself)
imply that $\mathbb{C}$ and $\mathbb{D}$ are complete lattices,
since the union of complete lattices is not necessarily a complete
lattice.

We do proceed, though, similarly%
\begin{comment}
TODO: Similar to which proof?
\end{comment}
{} to first prove that $\mathbb{C}$ is a meet-complete lattice. Particularly,
assuming $\mathbb{A}\subseteq\mathbb{C}$, we prove that $\bigwedge\mathbb{A}\in\mathbb{C}$.
Since $\mathbb{O}$ is a complete lattice, the point $\bigwedge\mathbb{A}$
exists in $\mathbb{O}$, and since $\mathbb{A}$ is a subset of $\mathbb{C}$
then the image set $F\left(\mathbb{A}\right)$ is a subset of $\mathbb{D}$
(by the definition of $\mathbb{D}$). By continuity, we have $F\left(\bigwedge\mathbb{A}\right)=\bigsqcap F\left(\mathbb{A}\right)$.
Also, let the set 
\[
\mathbb{B}=\left\{ B\in\mathbb{P}|\exists A\in\mathbb{A}.F\left(A\right)\sqsubseteq B\textrm{ and }G\left(B\right)\leq A\right\} 
\]
be the set of all points in $\mathbb{P}$ that form simultaneous pre-fixed
points when paired with some point in $\mathbb{A}$. Given that for
\emph{each} point $B\in\mathbb{B}$%
\begin{comment}
, as components of simultaneous pre-fixed points,
\end{comment}
{} there exists a point $A\in\mathbb{A}$ such that $F\left(A\right)\in F\left(\mathbb{A}\right)$
and $F\left(A\right)\sqsubseteq B$ (by the definition of simultaneous
fixed points), then the meet of all points $F\left(A\right)$ is less
than or equal to the meet of all points $B$, \emph{i.e.}, we have
$\bigsqcap F\left(\mathbb{A}\right)\sqsubseteq\bigsqcap\mathbb{B}$.
Using a similar argument, we also have $\bigwedge G\left(\mathbb{B}\right)\leq\bigwedge\mathbb{A}$.
By continuity, substituting for $\bigsqcap F\left(\mathbb{A}\right)$and
$\bigwedge G\left(\mathbb{B}\right)$ we have 
\[
F\left(\bigwedge\mathbb{A}\right)\sqsubseteq\bigsqcap\mathbb{B}\textrm{ and }G\left(\bigsqcap\mathbb{B}\right)\leq\bigwedge\mathbb{A}.
\]
As such, for all $\mathbb{A}\subseteq\mathbb{C}$, the point $\bigwedge\mathbb{A}$,
when paired with the point $\bigsqcap\mathbb{B}$, forms a pair of
simultaneous pre-fixed points of $F$ and $G$, and is thus a member
of $\mathbb{C}$. As such, $\mathbb{C}$ is a meet-complete lattice.
(Figure~\vref{fig:sim-prefix-latt-1} illustrates the proof for subsets
of $\mathbb{C}$ and $\mathbb{D}$ that have two elements.)
\begin{figure}
\noindent \begin{centering}
\includegraphics[scale=0.5]{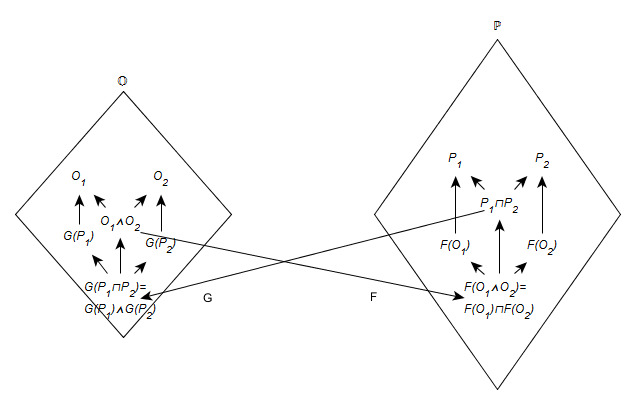}
\par\end{centering}
\caption{\label{fig:sim-prefix-latt-1}Illustrating Lemma~\ref{lem:PreFP-complatt}
(compare to Figure~\ref{fig:Cont-FG}).}
\end{figure}

Dually, we can prove that set $\mathbb{C}$ is also a join-complete
lattice (with the point $\top_{\mathbb{O}}$, the top element of $\mathbb{O}$,
at its top). Hence, set $\mathbb{C}$ is a meet-complete lattice and
a join-complete lattice, \emph{i.e.}, is a complete lattice.

Symmetrically, or by using Lemma~\ref{lem:cont-complatt}, we can
prove that set $\mathbb{D}$ also is a complete lattice (with the
point $\top{}_{\mathbb{P}}$, the top element of $\mathbb{P}$, at
its top). As such, the sets of all components of simultaneous pre-fixed
points are complete lattices.

Dually, we can also prove that the two sets $\mathbb{E}$ and $\mathbb{F}$
of all components of simultaneous \emph{post}-fixed points---\emph{i.e.},
duals of sets $\mathbb{C}$ and $\mathbb{D}$---are complete lattices
(with the points $\bot_{\mathbb{O}}$ and $\bot_{\mathbb{P}}$ at
their bottom), as required.
\end{proof}

\subsection{The Simultaneous Fixed Points Theorem}

The following theorem, asserting the existence of least and greatest
simultaneous fixed points, is the central theorem of this paper%
\begin{comment}
---one upon which the whole edifice of mutual (co)induction stands
\end{comment}
.\footnote{To the best of our knowledge, neither mutual (co)induction as we define
it in this paper nor a proof of the Simultaneous Fixed Points Theorem
have been presented formally before.}
\begin{thm}[The Simultaneous Fixed Points Theorem]
\label{thm:SFP}If $\left(\mathbb{O},\leq,\wedge,\vee\right)$ and
$\left(\mathbb{P},\sqsubseteq,\sqcap,\sqcup\right)$ are two complete
lattices and $F:\mathbb{O}\rightarrow\mathbb{P}$ and $G:\mathbb{P}\rightarrow\mathbb{O}$
are\emph{ }two continuous mutual endofunctions (i.e., two simultaneous
generators) over $\mathbb{O}$ and $\mathbb{P}$ then we have the
following:
\end{thm}
\begin{itemize}
\item the \emph{least simultaneous pre-fixed points }of $F$ and $G$, called
$\mu_{F}$ and $\mu_{G}$, exist in $\mathbb{O}$ and $\mathbb{P}$,
\item $\mu_{F}$ and $\mu_{G}$ are also the \emph{least simultaneous fixed
points} of $F$ and $G$,
\item the \emph{greatest simultaneous post-fixed points} of $F$ and $G$,
called $\nu_{F}$ and $\nu_{G}$, exist in $\mathbb{O}$ and $\mathbb{P}$,
and 
\item $\nu_{F}$ and $\nu_{G}$ are also the \emph{greatest simultaneous
fixed points} of $F$ and $G$.
\end{itemize}
\begin{proof}
Let the set 
\[
\mathbb{C}=\left\{ O\in\mathbb{O}|\exists P\in\mathbb{P}.F\left(O\right)\sqsubseteq P\textrm{ and }G\left(P\right)\leq O\right\} 
\]
 and the set 
\[
\mathbb{D}=\left\{ P\in\mathbb{P}|\exists O\in\mathbb{O}.F\left(O\right)\sqsubseteq P\textrm{ and }G\left(P\right)\leq O\right\} 
\]
 be the sets of all components of simultaneous pre-fixed points of
$F$ and $G$, and let points $\mu_{F}\textrm{ and }\mu_{G}$ be %
\begin{comment}
the least simultaneous pre-fixed points of $F\textrm{ and }G$ 
\end{comment}
defined as 
\begin{eqnarray}
\mu_{F} & = & \underset{F\left(O\right)\sqsubseteq P\textrm{ and }G\left(P\right)\leq O}{\bigwedge O}=\bigwedge\mathbb{C}\label{eq:muF}\\
\mu_{G} & = & \underset{F\left(O\right)\sqsubseteq P\textrm{ and }G\left(P\right)\leq O}{\bigsqcap P}=\bigsqcap\mathbb{D}.\label{eq:muG}
\end{eqnarray}

\begin{comment}
TODO: ---Revise FROM Here ---
\end{comment}
By Lemma~\ref{lem:PreFP-complatt}, points $\mu_{F}\textrm{ and }\mu_{G}$
are guaranteed to exist as the least elements of $\mathbb{C}$ and
$\mathbb{D}$. By the antisymmetry of $\leq$ and $\sqsubseteq$,
we can conclude that 
\begin{equation}
F\left(\mu_{F}\right)=\mu_{G}\textrm{ and }G\left(\mu_{G}\right)=\mu_{F}.\label{eq:least-prefix}
\end{equation}
First, we note that $F\left(\mu_{F}\right)\sqsubseteq\mu_{G}$ because
$\mu_{F}$ is the least element of $\mathbb{C}$ and thus, according
to the definition of simultaneous pre-fixed points (and as noted in
the proof of Lemma~\ref{lem:PreFP-img-complatt}), its image $F\left(\mu_{F}\right)$
is less than any element of $\mathbb{D}$, including the point $\mu_{G}$,
but, second, also we note that $\mu_{G}\sqsubseteq F\left(\mu_{F}\right)$
since $\mu_{G}$ is the least element of $\mathbb{D}$ and thus is
less than any point in $\mathbb{D}$, including the point $F\left(\mu_{F}\right)$.
By the antisymmetry of $\sqsubseteq$, we conclude that $F\left(\mu_{F}\right)=\mu_{G}$.

Symmetrically, we also have $G\left(\mu_{G}\right)=\mu_{F}$.

As such, the points $\mu_{F}$ and $\mu_{G}$ are simultaneous fixed
points of $F$ and $G$. They are also the least simultaneous pre-fixed
points of $F$ and $G$ since, by Equation~(\ref{eq:least-prefix}),
less-demandingly we have
\[
F\left(\mu_{F}\right)\sqsubseteq\mu_{G}\textrm{ and }G\left(\mu_{G}\right)\leq\mu_{F},
\]
meaning that points $\mu_{F}$ and $\mu_{G}$ are simultaneous pre-fixed
points of $F$ and $G$, and, by the individual uniqueness and minimality
of each of $\mu_{F}$ and $\mu_{G}$ (as the meets of the complete
lattices $\mathbb{C}$ and $\mathbb{D}$), points $\mu_{F}$ and $\mu_{G}$
are the least such points.

Now we have established both that $\mu_{F}\textrm{ and }\mu_{G}$
form the least simultaneous pre-fixed points of $F$ and $G$, and
that $\mu_{F}\textrm{ and }\mu_{G}$ are simultaneous fixed points
of $F$ and $G$, so $\mu_{F}\textrm{ and }\mu_{G}$ are the least
simultaneous fixed points of $F$ and $G$.

Using a dual argument, we can also prove that $\nu_{F}=\vee\mathbb{E}\textrm{ and }\nu_{G}=\sqcup\mathbb{F}$,
where sets $\mathbb{E}$ and $\mathbb{F}$ are the duals of sets $\mathbb{C}$
and $\mathbb{D}$ (see proof of Lemma~\ref{lem:PreFP-complatt}),
are the greatest simultaneous fixed points of $F\textrm{ and }G$,
as required.
\end{proof}

\section{\label{sec:Motivations-from-PLT}Motivations from PL Theory}

The significance of mutually-recursive definitions in programming
languages (PL) semantics and PL type theory is illustrated in the
following examples.

\subsection{Mutual Recursion at The Level of Data Values}

Lawrence Paulson, in his well-known book `ML for the Working Programmer'~\cite[p.58]{Paulson1996},
made some intriguing assertions, and presented an intriguing code
example. According to Paulson, ``Functional programming and procedural
programming are more alike than you may imagine''---a statement
that some functional programmers today are either unaware of, may
oppose, or may silently ignore. Paulson further states, verbatim,
that ``Any combination of \code{goto} and assignment statements
--- the worst of procedural code --- can be translated to a set
of mutually recursive functions.''

Then Paulson presents a simple example of imperative code. Here it
is.
\begin{lyxcode}
\begin{lstlisting}[language=C,tabsize=4]
var x := 0; y := 0; z := 0;
F:  x := x+1; goto G
G:  if y<z then goto F else (y := x+y; goto H)
H:  if z>0 then (z := z-x; goto F) else stop
\end{lstlisting}
\end{lyxcode}
To convert this imperative code into pure functional code, Paulson
then suggests: ``For each of the labels, \code{F}, \code{G}, and
\code{H}, declare mutually recursive functions. The argument of each
function is a tuple holding all of the variables.''

Here's the result when the method is applied to the imperative code
above:
\begin{lyxcode}
\begin{lstlisting}[language=ML,tabsize=4]
fun F(x,y,z) = G(x+1,y,z)  
and G(x,y,z) = if y<z then F(x,y,z) else H(x,x+y,z)
and H(x,y,z) = if z>0 then F(x,y,z-x) else (x,y,z);
\end{lstlisting}

\end{lyxcode}
\begin{comment}
and here's the typed and more structured version of the result:
\begin{lyxcode}
\begin{lstlisting}[language=ML,tabsize=4]
type int_trio = int * int * int

fun F(x,y,z: int): int_trio =
      G(x+1,y,z)  
and 
    G(x,y,z: int): int_trio =
      if y<z then F(x,y,z)
      else H(x,x+y,z)
and 
    H(x,y,z: int): int_trio =
      if z>0 then F(x,y,z-x)
	  else (x,y,z);
\end{lstlisting}
\end{lyxcode}
\end{comment}
Calling \code{F(0,0,0)} gives \code{x}, \code{y}, and \code{z}
their initial values for execution, and returns \code{(1,1,0)}---the
result of the imperative code. As such, Paulson concludes that: ``Functional
programs are referentially transparent, yet can be totally opaque''---
a statement which we read to mean that, for PL theorists in general
if not also for many mathematicians, functional programs (FP) are
{[}usually{]} easy to reason about, yet can {[}sometimes{]} be very
hard to reason about. Then Paulson concludes his discussion by sounding
the siren: ``If your code starts to look {[}sic{]} like this, beware!''

We can also introduce object-oriented programming (OOP) in this discussion.
In particular, a possible translation of Paulson's imperative code
to (non-imperative) OO code is as follows (in~\cite{AbdelGawad2017c}
we present a translation to a slightly more-succinct imperative OO
code):
\begin{lyxcode}
\begin{lstlisting}[language=Java,tabsize=4]
class C {
  final x, y, z: int
	
  // constructor
  C(xx,yy,zz: int) { x = xx; y = yy; z = zz }

  C F() { new C(x+1,y,z).G() }

  C G() { 
    if y < z then this.F()
    else new C(x,x+y,z).H() }

  C H() {
    if z > 0 then new C(x, y, z-x).F()
    else this }
}
\end{lstlisting}
\end{lyxcode}
Now, similar to Paulson's functional program, calling \code{new C(0,0,0).F()}
gives the fields \code{x}, \code{y}, and \code{z} their initial
values for execution, and returns an object equivalent to \code{new C(1,1,0)}---\emph{i.e.},
equivalent to the result of the imperative code.

Pondering a little over some of the ``worst'' imperative code and
over its translations to mutually-recursive functional and object-oriented
code suggests a strong similarity---if not equivalence---between
OOP, mutually-recursive FP, and procedural/imperative programming.
Further, this discussion implies that the mathematical-reasoning benefits
of functional programming---particularly the relative simplicity
of such reasoning---seem to crucially depend on \emph{not} heavily
using mutually-recursive function definitions (Paulson's concluding
warning can be read as an explicit warning \emph{against} writing
heavily mutually-recursive functional code). As the imperative code
and the OOP translation illustrate, and as is commonly known among
mainstream and industrial-strength software developers, however, heavily
mutually recursive definitions seem to be an essential and natural
feature of real-world/in\-dustrial-strength programming.

More significantly, the above translation between imperative, functional
and object-oriented code seems to also tell us that:

\begin{center}
\emph{if mutually-recursive functional programs can be reasoned about
mathematically,}\\
\emph{then also imperative and object-oriented programs (even the
worst such programs)}\\
\emph{can be reasoned about mathematically}.
\par\end{center}

A main objective behind the formal definition of mutual (co)induction
in this paper is to help in reasoning about mutually-recursive functional
programs mathematically (possibly making it even as simple as reasoning
about standard recursive functional programs is, based on using the
standard (co)induction principles), and, ultimately, to thereby possibly
help in reasoning about (even the worst?) imperative and object-oriented
programs mathematically too.

It should be noted that sometimes it is possible to reexpress mutual
recursion (also called \emph{indirect} recursion) or convert it into
standard recursion (also called \emph{direct} recursion), \emph{e.g.},
using the inlining conversion method presented in~\cite{Kaser1993,Yu1997}
or `with the help of an additional argument' as suggested by Paulson
also in~\cite[p.58]{Paulson1996}. From the results in~\cite{Kaser1993,Yu1997},
however, not all mutual recursion can be converted to standard recursion.
Extending from these results, while it is possible that some mutual
(co)inductive definitions can be converted to or encoded using standard
(co)inductive definitions, we conjecture that not all mutual (co)induction
can be translated into standard (co)induction. Hence, the need arises
for a genuine formal definition of mutual (co)induction that does
\emph{not} involve an encoding or translation of it into terms of
standard (co)induction.

\subsection{\label{subsec:Mutual-Recursion-Types}Mutual Recursion at The Level
of Data Types}

The previous section discussed mutual recursion at the level of values
(\emph{i.e.}, mutually-recursive data values, functions, or methods).
As is well-known in PL type theory, mutually-recursive types are essential
for typing mutually-recursive data values~\cite{MPS,TAPL}. Given
the ubiquity of mutually-recursive data values in OOP (via the special
variable \code{this/self}), mutually-recursive data \emph{types}
and mutually-recursive data \emph{type} constructors (\emph{e.g.},
\code{class}es, \code{interface}s, \code{trait}s, ...~etc.) are
ubiquitous in industrial-strength statically-typed OOP as well (\emph{e.g.},
in Java, C\#, C++, Kotlin, and Scala).

Further, in generic OOP languages such as Java, variance annotations
(such as wildcard types in Java) can be modeled by and generalized
to interval types~\cite{AbdelGawad2018c}. As presented in~\cite{AbdelGawad2018c},
the definition of ground types in Java depends on the definition of
interval types, whose definition, circularly (\emph{i.e.}, mutually-recursively),
depends on the definition of ground types. Further, the subtyping
relation between ground types depends on the containment (also called
\emph{subintervaling} or \emph{interval inclusion}) relation between
interval types, and vice versa. As such, the set of ground types and
the set of interval types in Java are examples of mutual recursive
sets, and the subtyping and the containment relations over these sets,
respectively, are mutually recursive relations too. (See~\cite{AbdelGawad2018b,AbdelGawad2018c}
for how, under an inductive interpretation, both sets---types and
interval types---and both relations---subtyping and containment---can
be iteratively constructed from the given set of classes of a Java
program together with the subclassing relation between these classes).

To illustrate all aspects of mutual recursion in OOP (\emph{i.e.},
at the level of type, at the level of values, and in defining types/subtyping
and interval types/containment), the following OO code, written in
an imaginary Java-like language, presents a simple set of mutually
recursive classes to model the mutual recursion between in the definition
of types and interval types.
\begin{lyxcode}
\begin{lstlisting}[language=Java,tabsize=4]
class Class {
  String name; // holds the name of the class
}

class Type {
  Class c;
  Interval i; // the type arg of c. null if c is not generic
}

class Interval {
  Type UBnd; // the upper bound of the interval type
  Type LBnd; // its lower bound
}
// Note the mutual recursion (at the level of types) between
//   the definitions of classes Type and Interval
\end{lstlisting}
\end{lyxcode}
And the following code adds to the code above a simple set of mutually
recursive methods to model the mutual recursion in the definitions
of the subtyping and containment relations.
\begin{lyxcode}
\begin{lstlisting}[language=Java,tabsize=4]
class Class {
  String name;

  bool isSubclassOf(Class c) {
    // Handle special classes Null and Object
    if(this == NullCls || c == ObjCls) return true;
    // Else check if this class inherits from class c
    return inher_table.lookup(this, c);
  }
}
\end{lstlisting}

\begin{lstlisting}[language=Java,tabsize=4]
class Type {
  Class c;
  Interval i;

  bool isSubtypeOf(Type t){
    // assuming that i and t.i are not null,
    //   ie, that c and t.c are generic
    return c.isSubclassOf(t.c) && i.isSubintOf(t.i)
  }
}
\end{lstlisting}

\begin{lstlisting}[language=Java,tabsize=4]
class Interval {
  Type UBnd;
  Type LBnd;

  bool isSubint(Interval i){
    // covar in upperbound, contravar in lower bound
    return UBnd.isSubtypeOf(i.UBnd) && i.LBnd.isSubtypeOf(LBnd)
  }
}
// Note the mutual recursion (at the level of values) between
//   the definitions of methods isSubtypeOf and isSubintOf
\end{lstlisting}
\end{lyxcode}

\subsection{Mutual Coinduction in OOP}

As discussed in detail in $\mathsection$\ref{subsec:Mutual-Recursion-Types},
mutually recursively definitions exist in OOP at two levels: at the
level of values (%
\begin{comment}
open (co?)recursion, 
\end{comment}
via \code{this}) and at the level of types (\emph{i.e.}, between
classes, and in the definition of the subtyping and containment relations).
It should be noted that the subtyping relation is covariant/monotonic
w.r.t. containment, and that containment is covariant w.r.t. subtyping
in the first argument (\emph{i}.\emph{e.}, w.r.t. the upper bound
of an interval type) and is contravariant w.r.t. the second argument
(\emph{i.e.}, w.r.t. the lower bound of an interval type). As such,
in Java with interval types, subintervals generates subtypes, and
subtypes in the upper bounds of interval types generates subintervals,
while subtypes in the lower bounds of interval types generates superintervals.
\begin{comment}
TODO: type negation involved? Type intervals~\cite{CorkyIntervals}?
\end{comment}
.

Just as for standard (co)induction, where the notions of (least) pre-fixed
points and (greatest) post-fixed points have relevance and practical
value even when such points do \emph{not }correspond to fixed points,
\emph{e.g.}, when the underlying posets are \emph{not} complete lattices
or when the generators are \emph{not} monotonic~(see, for example,~\cite{AbdelGawad2019b,AbdelGawad2018f}),
we expect that the notions of (least) simultaneous pre-fixed points
and (greatest) post-fixed points to have relevance and practical value
even when these points do \emph{not} correspond to simultaneous fixed
points, \emph{e.g.}, when the underlying posets are \emph{not} complete
lattices or when the generators are \emph{not} continuous (and may
not be even monotonic).

\section{\label{sec:Requiring-Only-Monotonicity}Requiring Only Monotonicity}

While the proofs presented in Appendix~\ref{sec:Theorems-and-Proofs}
are correct, and the continuity of simultaneous/mutual generators
does indeed guarantee the existence of least and greatest simultaneous
fixed points in complete lattices, but, as we suggested in Section~\ref{sec:Conclusion-and-FW},
in this appendix we explain how indeed monotonicity of the generators
alone is enough to prove the existence of least and greatest simultaneous
fixed points, thereby relaxing the condition on generators.

The first hint that monotonicity is enough (that in fact is not just
a hint but is a proof, albeit an indirect one) comes from the fact
that monotonic mutual generators over complete lattices can be represented
or encoded, in a standard way, as one monotonic generator over one
complete lattice, for which a (standard) least and greatest fixed
point is guaranteed to exist (by the standard Knaster-Tarski theorem),
and in which the least fixed point and the greatest fixed point correspond
directly to the least simultaneous fixed point and the greatest simultaneous
fixed point of the original mutual generators.

In particular, as in Appendix~\ref{sec:Theorems-and-Proofs}, let
$F:\mathbb{O}\rightarrow\mathbb{P}$ and $G:\mathbb{P}\rightarrow\mathbb{O}$
be two monotonic mutual generators over complete lattices $\mathbb{O}$
and $\mathbb{P}$. Then the product $\mathbb{OP}\doteq\mathbb{O}\times\mathbb{P}$
\emph{is} also a complete lattice under the standard component-wise
ordering (i.e., where $\left(o_{1},p_{1}\right)\le\left(o_{2},p_{2}\right)\mathrm{\,iff\,}o_{1}\le o_{2}\wedge p_{1}\sqsubseteq p_{2}$).
Further, the function $H\left(o,p\right):\mathbb{OP}\rightarrow\mathbb{OP}\doteq\left(G\left(p\right),F\left(o\right)\right)$
\emph{is} a monotonic function over $\mathbb{OP}$ (given the monotonicity
of $F$ and $G$). As such, by the standard Knaster-Tarski theorem,
function $H$ is a generator over the complete lattice $\mathbb{OP}$
that has a least fixed point $\mu_{\mathbb{OP}}=\left(\mu_{\mathbb{O}},\mu_{\mathbb{P}}\right)$
and a greatest fixed point $\nu_{\mathbb{OP}}=\left(\nu_{\mathbb{P}},\nu_{\mathbb{P}}\right)$.
Given our definition of a simultaneous pre-fixed point $\left(o,p\right)$
of $F$ and $G$ (as one where $G\left(p\right)\le o\wedge F\left(o\right)\sqsubseteq p$)
implies that $\mu_{\mathbb{OP}}$(the least fixed point of $H$) is
a simultaneous pre-fixed point of $F$ and $G$, is the least such
point, and, thus, is also the least simultaneous fixed point of $F$
and $G$. Similarly, we can conclude that $\nu_{\mathbb{OP}}$ (the
greatest fixed point of $H$) is the greatest simultaneous fixed point
of $F$ and $G$.

It should be noted, though, that while this proof does not require
the continuity of $F$ and $G$, it only indirectly constructs the
least and greatest simultaneous fixed points of $F$and $G$ as decodings
of the least and greatest fixed points of another function ($H$).
Below we present a more direct proof of the existence of these simultaneous
fixed points that does \emph{not} depend on the conclusion of the
standard Knaster-Tarski theorem, and that also discusses the issue
of multiplicity of pairings in simultaneous pre-fixed and simultaneous
post-fixed points (which we discussed in Section~\ref{sec:Mutual-(Co)Induction},
and led us to consider using continuity). We believe our next proof,
while more intricate and involved that the encoding proof, does reveal
a bit more about the structure and the inner workings of mutual (co)induction
than a standard encoding of mutual (co)induction in terms of standard
(co)induction does.%
\begin{comment}
And Allah SWT (God) knows best.
\end{comment}
{} %
\begin{comment}
TODO: Present argument that requires monotonicity but not continuity.
In particular, Lemmas 4, 5, and 6 above will fail ... yet, using monotonicity,
we can still prove the existence of lsfp and gsfp.
\end{comment}

As an appetizer and source of intuition, let's first note the following
lemma.
\begin{lem}[The lub (glb) of components of simultaneous pre(post)-fixed (post-fixed)
points forms a simultaneous pre(post)-fixed point.]
\label{lem:mono-lub-glb}Given the definitions of Theorem~(\ref{thm:SFP}),
if $\left(o,p_{1}\right)$ and $\left(o,p_{2}\right)$ are two simultaneous
pre-fixed points then, by the monotonicity of G, $\left(o,p_{1}\sqcap p_{2}\right)$
is a simultaneous pre-fixed point, and, dually, if $\left(o,q_{1}\right)$
and $\left(o,q_{2}\right)$ are two simultaneous post-fixed points
then, by monotonicity of $F$, $\left(o,q_{1}\sqcup q_{2}\right)$
is a simultaneous post-fixed point.

Further, if $\left(o_{1},p_{1}\right)$ and $\left(o_{2},p_{2}\right)$
are two simultaneous pre-fixed points then, by the monotonicity of
G, $\left(o_{1}\wedge o_{2},p_{1}\sqcap p_{2}\right)$ is a simultaneous
pre-fixed point, and, dually, if $\left(o_{1},q_{1}\right)$ and $\left(o_{2},q_{2}\right)$
are two simultaneous post-fixed points then, by the monotonicity of
$F$, $\left(o_{1}\vee o_{2},q_{1}\sqcup q_{2}\right)$ is a simultaneous
post-fixed point as well.
\end{lem}
\begin{proof}
We have $G\left(p_{1}\right)\le o$ and $F\left(o\right)\sqsubseteq p_{1}$
and $G\left(p_{2}\right)\le o$ and $F\left(o\right)\sqsubseteq p_{2}$,
which means $F\left(o\right)$ is a lower bound of $p_{1}$ and $p_{2}$
but $p_{1}\sqcap p_{2}$ is the greatest lower bound and hence $F\left(o\right)\sqsubseteq p_{1}\sqcap p_{2}$
(1). By the monotonicity of $G$, $p_{1}\sqcap p_{2}\sqsubseteq p_{1}\implies G\left(p_{1}\sqcap p_{2}\right)\le G\left(p_{1}\right)$,
and thus, by transitivity, we have $G\left(p_{1}\sqcap p_{2}\right)\le o$
(2). Combining (1) and (2), the pair $\left(o,p_{1}\sqcap p_{2}\right)$
is a simultaneous pre-fixed point too. (Trying to do this for $p_{1}\sqcup p_{2}$
fails since monotonicity does not necessitate that $G\left(p_{1}\sqcup p_{2}\right)\le o,$but
continuity does ... and to prove the existence of a least simultaneous
fixed point this extra requirement is actually%
\begin{comment}
it seems so
\end{comment}
{} not needed).

Dually, the monotonicity of $G$ also can be used to prove that $o\le G\left(q_{1}\right)$
and $q_{1}\sqsubseteq F\left(o\right)$ and $o\le G\left(q_{2}\right)$
and $q_{2}\sqsubseteq F\left(o\right)$ implies that $q_{1}\sqcup q_{2}\sqsubseteq F\left(o\right)$
($F\left(o\right)$ is greater than or equal the least upper bound
of $q_{1}$ and $q_{2}$) and that $o\le G\left(q_{1}\right)\le G\left(q_{1}\sqcup q_{2}\right)$
and, hence, that $\left(o,q_{1}\sqcup q_{2}\right)$ is a simultaneous
post-fixed point.
\end{proof}
Similarly, by the monotonicity of $F$, if $\left(n_{1},p\right)$
and $\left(n_{2},p\right)$ are simultaneous pre-fixed points then
$\left(n_{1}\wedge n_{2},p\right)$ is a simultaneous pre-fixed point
(but not necessarily $\left(n_{1}\vee n_{2},p\right)$), and if $\left(o_{1},p\right)$
and $\left(o_{2},p\right)$ are simultaneous post-fixed points then
$\left(o_{1}\vee o_{2},p\right)$ is a simultaneous post-fixed point
(but not necessarily $\left(o_{1}\wedge o_{2},p\right)$).

Given the intuitions provided by Lemma~\ref{lem:mono-lub-glb} and
its dual we can now prove, more directly, the existence of a lsfp
and gsfp as follows. In particular, we prove, now using monotonicity
alone, that the points $\mu_{F}\in\mathbb{O}$ and $\mu_{G}\in\mathbb{P}$
defined in Equations~(\ref{eq:muF}) and~(\ref{eq:muG}) on page~\pageref{eq:muG}
form the least simultaneous fixed point. We do so in two steps, where
we first prove that $\mu_{F}$ and $\mu_{G}$ form the least simultaneous
pre-fixed point, then that they form a simultaneous fixed point, from
which we conclude that $\mu_{F}$ and $\mu_{G}$ define the least
simultaneous fixed point of $F$ and $G$ (only assuming the monotonicity
of $F$ and $G$ but not requiring their continuity).

To prove that $\mu_{F}$ and $\mu_{G}$ define a simultaneous pre-fixed
point first let's note that by the monotonicity of $G$ \footnote{Note that the greatest lower bound (glb) of any set is less than or
equal any element of the set, and thus applying a monotonic function
to the glb produces an element less than or equal to the value of
the function at any element of the set, and is thus a lower bound
of the set of function values. In particular, $p_{i}\leq p_{j}\implies G\left(p_{i}\right)\leq G\left(p_{j}\right)$
implies that $G\left(\sqcap p_{i}\right)\leq G\left(p_{i}\right)$
for each and all $p_{i}$ .}, we have
\[
G\left(\mu_{G}\right)=G\left(\sqcap\mathbb{D}\right)\leq G\left(P_{i}\right)\mathrm{\,for\,all\,}P_{i}\in\mathbb{D}.
\]
 Now, from the definition of simultaneous pre-fixed points we have
$G\left(P_{i}\right)\leq O_{i}\mathrm{\,for\,all\:O_{i}\in\mathbb{C}}$.
As such, by the transitivity of $\leq$, we have $G\left(\mu_{G}\right)\leq O_{i}\mathrm{\,for\,all\:O_{i}\in\mathbb{C}}$.
Thus, $G\left(\mu_{G}\right)$ is a lower bound of $\mathbb{C}.$Given
that $\mu_{F}$is the greatest lower bound of $\mathbb{C}.$thus we
have $G\left(\mu_{G}\right)\leq\mu_{F}$. Similarly, we can prove
that $F\left(\mu_{F}\right)\sqsubseteq\mu_{G}$, thereby proving that
$\mu_{F}$ and $\mu_{G}$ define a simultaneous pre-fixed point.

Further. the definitions of $\mu_{F}$ and $\mu_{G}$ as the \emph{least}
elements of $\mathbb{C}$ and $\mathbb{D}$, respectively, show that
they define the \emph{least} simultaneous pre-fixed point.

Now, following a reasoning much similar to that used in the proof
of the standard Knaster-Tarski fixed point theorem (but not using
its conclusion) (e.g., see~\cite[p.283]{TAPL}), we can prove that
$\mu_{F}$ and $\mu_{G}$ also define a simultaneous \emph{fixed}
point.

As such, given that each simultaneous fixed point is a simultaneous
pre-fixed point, then combining the conclusions of both steps lets
us conclude that $\mu_{F}$ and $\mu_{G}$ define the \emph{least}
simultaneous \emph{fixed} point (lsfp).

Using a dual argument we can prove, only assuming the monotonicity
of $F$ and $G$, that points $\nu_{F}\in\mathbb{O}$ and $\nu_{G}\in\mathbb{P}$
(as defined in Theorem~\ref{thm:SFP}) define the greatest simultaneous
fixed point (gsfp) of $F$ and $G$.
\end{document}